\documentclass[acmsmall,screen,authorversion,nonacm]{acmart}
\startPage{1}

\setcopyright{rightsretained}
\copyrightyear{2018}           %

\usepackage{booktabs}   %
\usepackage{subcaption} %
\usepackage[frozencache]{minted}
\usepackage{proof,mathtools}
\usepackage{kbordermatrix,xcolor,colortbl}
\usepackage{enumitem}
\usepackage{siunitx}
\setcitestyle{round}
\usepackage{xspace}

\newcommand\nth{\textsuperscript{th}\xspace}

\setminted{fontsize=\fontsize{9}{10}}

\fvset{%
  xleftmargin=1em
}

\theoremstyle{definition}

\DeclarePairedDelimiterX\setc[2]{\{}{\}}{\,#1 \;\delimsize\vert\; #2\,}

\def\VR{\kern-\arraycolsep\strut\vrule &\kern-\arraycolsep}
\def\vr{\kern-\arraycolsep & \kern-\arraycolsep}
\newcommand\inlm[1]{\text{\mintinline{fortran}{#1}}}
\newcommand\inlms[1]{\text{\mintinline[fontsize=\footnotesize]{fortran}{#1}}}
\newcommand\inl{\mintinline{fortran}}

\newcommand\inlFS{\mintinline{fsharp}}
\newcommand\prn[1]{\left({#1}\right)}
\newcommand\abs[1]{\left|{#1}\right|}
\newcommand\len[1]{\left|{#1}\right|}
\newcommand\set[1]{\left\{{#1}\right\}}

\newcommand\punitA{{\tt 'a}}
\newcommand\punitB{{\tt 'b}}

\newcommand\fsharp{F\#}
\newcommand\sub[2]{\ensuremath{{{#1}}\!\prn{#2}}}

\newcommand\coneq{\sim}

\newcommand\UName{U}
\newcommand\Uvar[1]{\UName\!\left[{#1}\right]}
\newcommand\Ulit[1]{\UName\!\left[{#1}\right]}

\newcommand\Uabs[2]{\Uuse{#1}{#2}{\ast}}
\newcommand\UEAPabs[1]{\UName\!\left[{#1},{\ast}\right]}
\newcommand\Uuse[3]{\UName\!\left[{#1},{#2},{#3}\right]}
\newcommand\UEAPuse[2]{\UName\!\left[{#1},{#2}\right]}

\newcommand\rulenamed[1]{\ensuremath{{#1}}}
\newcommand\erule[3]{\infer[{\mathit{#1}}_{\mathit{ex}}]{#2}{#3}}
\newcommand\erulenamed[1]{\rulenamed{{\mathit{#1}}_{\mathit{ex}}}}
\newcommand\strule[3]{\infer[{\mathit{#1}}_{\mathit{st}}]{#2}{#3}}
\newcommand\strulenamed[1]{\rulenamed{{\mathit{#1}}_{\mathit{st}}}}
\newcommand\purule[3]{\infer[{\mathit{#1}}_{\mathit{pu}}]{#2}{#3}}

\newcommand\blrule[3]{\infer[{\mathit{#1}}_{\mathit{bl}}]{#2}{#3}}
\newcommand\blrulenamed[1]{\rulenamed{{\mathit{#1}}_{\mathit{bl}}}}

\newcommand\cstep[5]{{{#1};{#2}}\vdash{#3}\leadsto{{#4};{#5}}}
\newcommand\ststep[4]{{{#1};{#2}}\vdash{#3}\leadsto{{#4}}}
\newcommand\pustep[4]{{{#1};{#2}}\vdash{#3}\leadsto{{#4}}}
\newcommand\modstep[3]{{{#1}}\vdash{#2}\leadsto{{#3}}}
\newcommand\ujudge[4]{{#1}\vdash{#2}:{#3};{#4}}

\newcommand\polycontext[1]{\text{polycontext}\!\prn{#1}}
\newcommand\freshcallid[1]{\text{freshid}\!\prn{#1}}
\newcommand\freshlitid[1]{\text{freshid}\!\prn{#1}}

\newcommand\isfunction[1]{\text{isfunction}\!\prn{#1}}
\newcommand\isarray[1]{\text{isarray}\!\prn{#1}}
\newcommand\isMonomorphic[1]{\text{isMonomorphic}\!\prn{#1}}
\newcommand\instantiate[2]{\text{instantiate}\!\prn{{#1},{#2}}}
\newcommand\unitless{{\mathbf{1}}}

\newcommand\external{{\mathsf{external}}}
\newcommand\Array{{\mathsf{array}}}

\newcommand\conv[1]{\left\llbracket{#1}\right\rrbracket}
\newcommand\lhs{{\mathsf C}}
\newcommand\rhs{{\mathsf B}}
\newcommand\unknownv[1]{\lhs\!\left[{#1}\right]}
\newcommand\baseunit[1]{\rhs\!\left[{#1}\right]}

\definecolor{darkgrey}{rgb}{0.5,0.5,0.5}
\definecolor{darkgreen}{rgb}{0.0,0.5,0.0}
\definecolor{darkpurple}{rgb}{0.6,0.0,0.6}
\definecolor{orange}{rgb}{0.8,0.4,0.0}

\usepackage{fancyvrb}
\RecustomVerbatimEnvironment
{Verbatim}{Verbatim}
{fontsize=\small}

\usepackage{url}

\begin{document}

\title{Incremental units-of-measure verification} %
\author{Matthew Danish}
\affiliation{
  \position{Research Associate}
  \department{Department of Computer Science and Technology}      %
  \institution{University of Cambridge} %
  \streetaddress{15 JJ Thomson Ave}
  \city{Cambridge}
  \state{Cambridgeshire}
  \postcode{CB3 0FD}
  \country{United Kingdom}                    %
}
\email{mrd45@cam.ac.uk}          %

\author{Dominic Orchard}
\affiliation{
  \position{Lecturer}
  \department{School of Computing}      %
  \institution{University of Kent} %
  \streetaddress{}
  \city{Canterbury}
  \state{Kent}
  \postcode{CT2 7NF}
  \country{United Kingdom}                    %
}
\email{D.A.Orchard@kent.ac.uk}          %

\author{Andrew Rice}
\affiliation{
  \position{Reader}
  \department{Department of Computer Science and Technology}      %
  \institution{University of Cambridge} %
  \streetaddress{15 JJ Thomson Ave}
  \city{Cambridge}
  \state{Cambridgeshire}
  \postcode{CB3 0FD}
  \country{United Kingdom}                    %
}
\email{acr31@cam.ac.uk}          %

\begin{abstract}
Despite an abundance of proposed systems, the verification of
units-of-measure within programs remains rare in scientific
computing. We attempt to address this issue by providing a lightweight
static verification system for units-of-measure in Fortran programs which supports incremental
annotation of large projects. We take the opposite approach to the most mainstream
existing deployment of units-of-measure typing (in \fsharp) and generate a
global, rather than local, constraints system for a program. We show that such a system can infer (and check)
polymorphic units specifications for under-determined parts of the
program. Not only does this ability allow checking of partially annotated
programs but it also allows the global constraint
problem to be partitioned. This partitioning means we can scale to large programs
by solving constraints for each program module independently and
storing inferred units at module boundaries (\emph{separate verification}). We provide an
implementation of our approach as an extension to an open-source
Fortran analysis tool.
\end{abstract}




\maketitle

\section{Introduction}

Scientific computing and computational models are playing an
increasingly important role in research, industry, and policy. However, such programs
are just as prone to error as other kinds of software and
their reliance on complex numerical routines makes testing difficult.
Very simple errors, such as a flipped minus sign, have had significant
impact, for example, leading to retractions from premier journals~\citep{Merali:2010}.
Lightweight verification techniques have the potential to help in this
situation~\citep{ORCHARD2014713} by providing scientists with tools
that have a low specification burden and integrate with existing
practices. One example of this is dimension typing in which the consistency of
dimensionality or units-of-measure in equations is statically verified.
Units-of-measure mistakes have led to high-stakes disasters, such as
the \$327 million Mars Climate Orbiter~\citep{MarsClimateOrbiter}
disappearing during orbital insertion. An investigation later pinned
the blame on a mix-up of units-of-measure, with ground-based software
sending numbers in Imperial units instead of metric units, dropping the
probe into the atmosphere of Mars~\citep{Stephenson:1999:MCO}. This is
not the only time NASA projects have struggled to convert from Imperial
to metric units. Such errors are implicated in at least one other
crash~\citep{Marks:2009}.

Checking the consistency of units-of-measure in equations, or at least
the physical dimensions of calculations, is a method long employed by
scientists working with pen and paper~\citep{Macagno:1971}. Informal
discussion with scientists and examination of scientific program code
has shown that units-of-measure reasoning in code does take place but
usually only by hand~\citep{Orchard:2015:JCS}.
Formal units checking has not been adopted in scientific
programming despite the familiarity of the technique and a multitude
of proposed systems for automated
verification~\citep{Ore:2017:ISSTA,Jiang:2006:ICSE,Kennedy:2009:CEFP,Snyder:2013:N1969}. One
potential explanation for this is the burden of applying units
annotations to existing codebases, which may be a large burden if many
annotations are needed (some approaches even require that every
variable declaration is annotated).

In this paper, we describe a fresh approach to statically analysing and verifying
units-of-measure which eases this burden by supporting incremental
annotation, rather than requiring a whole program change.  Units
annotations (specifications) are orthogonal to the type system and
placed in comments. Thus, we integrate smoothly with existing tools
and compilers, without modifying the language or
relying on the existing type system. We support any
system of units that the user would like to use, not only absolute
units. We also allow user-defined convenient aliases for units.

The most mainstream units-of-measure checking system is part of the programming
language \fsharp\ using an extension to its type
system~\citep{Kennedy:2009:CEFP} that uses local reasoning based on
Kennedy's extension to ML~\citep{Kennedy:1996} to support unification
over the equational theory of Abelian groups. Our principal
contribution is an algorithm that takes a different approach to the same
problem, by generating and solving global constraints. We describe the
core algorithm and its implications in Section~\ref{sec:algorithm}.
We show how our global constraint system can infer polymorphic units
for underdetermined systems (Section~\ref{sec:backendpoly}) and that
this supports both genuine polymorphic functions
(Section~\ref{sec:polyinst}) and also partial annotation of programs
(Section~\ref{sec:incremental}). We also show that dividing a large
program into modules and recording inferred polymorphic units at the
boundaries partitions the global constraint problem and allows our
approach to scale to large programs (Section~\ref{sec:fsmod}). We give
a quantitative analysis of this scaling in Section~\ref{sec:scale}.

Our approach is based on a novel modification of standard techniques
for reducing a matrix of constraints to Hermite Normal
Form. The modified algorithm allows us to choose a single
solution for a set of polymorphic unit constraints that would
otherwise form an underconstrained system if simply viewed as a set of
linear equations. We formalise our algorithm, prove
termination and prove that it produces integer solutions
(Section~\ref{sec:modhnf}).

We implemented our approach as an extension of CamFort, an open-source Fortran analysis and verification
tool~\citep{CamFort}. However, our work is
applicable beyond Fortran. We highlight in Section~\ref{sec:algorithm} any
Fortran-specific parts and how our approach can
be more generally applied.
Section~\ref{sec:studies} details practical experiences applying our
tool to real-world scientific computing models; we
discuss several examples of bugs and undocumented unit conversions
found by our tool.
Section~\ref{sec:related} compares our approach against the
literature, including the work of
\citet{Kennedy:2009:CEFP,Kennedy:1996,kennedy1994dimension},
\citet{Jiang:2006:ICSE},
\citet{Hills:2008:RULE,Chen:2003:RTA,Rosu:2003:ASE},
\citet{Orchard:2015:JCS,Contrastin:2016},
\citet{Ore:2017:ISSTA},
\citet{Squants} and others.

\section{Overview}
\label{sec:overview}

\subsection{Annotating and Checking Units-of-Measure}\label{sec:suggest}

A physical dimension is a measured quantity such as length, mass, time
or electric charge that can be sensibly compared among its own
kind. Other dimensions are derived, such as velocity
measured as a length over time. Dimensional analysis is a standard
technique for checking calculations and understanding physical
equations: by tracing the steps of computation with labelled
dimensions one can check if the resulting dimension is as
intended. \emph{Units-of-measure} are the conventional names adopted for
specific amounts of each dimension. For example, the length dimension
can be expressed in terms of metres, inches, feet,
Smoots~\citep{Smoot}, among many others.

\newcommand\nonterm[1]{{\emph{#1}}}
\newcommand\term[1]{{\text{\sf #1}}}
\newcommand\bndef{::=}
\newcommand\synalt{\phantom{\bndef}\mid}

\begin{listing}[b]
\vspace{-0.35em}
\begin{minipage}{0.45\linewidth}
\begin{minted}{fortran}
program ballistics
  implicit none
  real, parameter :: x0 = 0
  real, parameter :: v0 = 20
  real, parameter :: a = -9.8
  real :: x, t
  x = 0.5 * a * t * t + v0 * t + x0
end program ballistics
\end{minted}
\end{minipage}
\hspace{0.5em}
\begin{minipage}{0.45\linewidth}
\begin{Verbatim}
$ camfort units-suggest ballistics.f90 

ballistics.f90: 3 variable declarations
   suggested to be given a specification:
    (5:22)    a
    (4:22)    v0
    (6:11)    x
\end{Verbatim}
\vspace{1em}
\end{minipage}
\vspace{-0.55em}
\caption{Transcription of ballistics equation: un-annotated (left)
  with suggestion from our tool (right).}
\label{lst:ball1}
\end{listing}

As a simple example, consider the elementary equation of ballistics:
$x(t)=\frac 1 2 at^2+v_0t+x_0$ transcribed into Fortran in
Listing~\ref{lst:ball1} (left).  We could annotate every variable using our
comment-based syntax, but to save work we first ask our system to
suggest a small set of variables that is sufficient to determine the
others, by the command %
\texttt{camfort units-suggest} (Listing~\ref{lst:ball1}, right).

In this case, our tool suggests annotating \inl{a}, \inl{v0} and \inl{x}.
We annotate those variables as shown in Listing~\ref{lst:infer}
(left) using our comment-based unit syntax (define in
Figure~\ref{fig:syntax}). We can then invoke the inference engine
as shown in Listing~\ref{lst:infer} (right). The result is that every variable
has been successfully assigned a consistent units-of-measure.  We can
import these inferred units back into the program text using a related
command, {\tt units-synth}, which rewrites the original source code to a
specified file but with the inferred annotations automatically
inserted as comments, as a time-saving measure. Those annotations will
then serve as static specifications that can be checked any time in
the future, for reassurance, and to aid code maintenance and documentation.

\begin{figure}[b]
  \centering
  \begin{align*}
    \nonterm{annotation}&\bndef\ \nonterm{specification} \mid \nonterm{alias} &
    \nonterm{polyname} &\bndef\ \term{'} \nonterm{name} \\
    \nonterm{specification} &\bndef\ \nonterm{commentchar}~\term{=}~\term{unit}~\nonterm{u}~\term{::}~\nonterm{vars}&
    \nonterm{name} &\bndef\ \term{[a-zA-Z][a-zA-Z0-9\_-]*} \\
    \nonterm{alias}&\bndef\ \nonterm {commentchar}~\term{=}~\term{unit}~\term{::}~\nonterm{name}~\term{=}~\nonterm{u}&
    \nonterm{vars} &\bndef\ \nonterm{vars} , \nonterm{var} \mid \nonterm{var} \\
    \nonterm{commentchar} &\bndef\ \term{!} \prn{Fortran90} \mid \term c \prn{Fortran77} &
    \nonterm{var} &\bndef\ \term{[a-zA-Z][a-zA-Z0-9\_-]*}\\
    \nonterm{u} &\bndef\ \nonterm{$U_{\textit{name}}$} \mid\unitless \mid uu \mid u/u \mid u^z&
    \nonterm{z} &\in\mathbb{Z}\\
    \nonterm{$U_{\textit{name}}$} &\bndef\ \nonterm{polyname} \mid \nonterm{name} \\
  \end{align*}
  \vspace{-3em}
  \caption{Annotation syntax: units specification and aliases}
  \label{fig:syntax}
\end{figure}

\begin{listing}[t]
\begin{minipage}{0.45\linewidth}
\begin{minted}{fortran}
program ballistics
  implicit none
  real, parameter :: x0 = 0
  != unit metre / sec :: v0
  real, parameter :: v0 = 20
  != unit metre / sec**2 :: a
  real, parameter :: a = -9.8
  != unit metre :: x
  real :: x, t
  x = 0.5 * a * t * t + v0 * t + x0
end program ballistics
\end{minted}
\end{minipage}
\hspace{1em}
\begin{minipage}{0.45\linewidth}
\begin{Verbatim}
$ camfort units-infer ballistics.f90
ballistics.f90:
  3:22 unit metre :: x0
  5:22 unit metre / sec :: v0
  7:22 unit metre / (sec**2) :: a
  9:11 unit metre :: x
  9:14 unit sec :: t
\end{Verbatim}
\vspace{3em}
\end{minipage}
\caption{\textit{Left:} transcription of the ballistics equation with suggested
variables \inl{a}, \inl{v0} and \inl{x}
annotated with units. \textit{Right}: output of units-infer mode at command-line applied
to the partially annotated program, showing line
  number and column of inferred variables, e.g., at line 3 column 22,
\inl{x0} is inferred to have unit \inl{metre}.}
\label{lst:infer}
\vspace{-0.5em}
\end{listing} %

\subsection{Polymorphic Units Annotation}

\begin{listing}[t]
\hspace{4em}
\begin{minipage}{0.4\linewidth}
\begin{minted}{fortran}
program double
  != unit metre :: x
  != unit sec   :: t
  real :: x = 10, t = 15
  x = d(x)
  t = d(t)

contains

  != unit 'a :: d
  real function d(n)
    != unit 'a :: n
    real :: n
    d = 2 * n
  end function d
end program double
\end{minted}
  \vspace{1em}
\end{minipage}
\hspace{1em}
\begin{minipage}{0.4\linewidth}
\begin{minted}{fortran}
program square
  != unit metre    :: x
  != unit metre**2 :: y
  real :: x = 10, y
  != unit sec      :: t
  != unit sec**2   :: s
  real :: t = 5, s
  y = sqr(x)
  s = sqr(t)
contains
  != unit 'a**2 :: sqr
  real function sqr(n)
    != unit 'a :: n
    real :: n
    sqr = n * n
  end function sqr
end program square
\end{minted}
\end{minipage}
  \caption{Annotated examples of polymorphism.}
  \label{lst:double}\label{lst:square}
\end{listing}

\noindent
Many functions (and subroutines) can correctly operate in a generic
way on any units they are given.  Listing~\ref{lst:double} gives two
simple examples of units-polymorphic functions. Function~\inl{d(n)}
simply doubles the value of its parameter and
function~\inl{square(n)} squares the value of its parameter. All
of the units are annotated in this example. Variables \inl{x} and
\inl{t} are given concrete units \inl{metre} and \inl{sec}. But in the
function \inl{d}, the parameter \inl{n} is annotated with units named
{\tt 'a}. This is a \emph{polymorphic} units variable that can
be %
substituted with any concrete units as needed. We follow the ML
tradition, where polymorphic names start with a single-quote and their
names are pronounced by their Greek alphabet equivalent, e.g. $\alpha$
for {\tt 'a}. The same variable is used to annotate \inl{d} (its
result) indicating that for every call to \inl{d} the units of the
function's result are the units of any argument substituted for \inl{n}.
This can be seen in the program where both %
\inl{x = d(x)} and \inl{t = d(t)} are valid lines: one instantiating
\inl{d} for units \inl{metre} and the other for \inl{sec}.
Note that this implies that the literal \inl{2} is inferred to
be \emph{unitless} (a scaling factor).

A similar scheme of polymorphic unit annotations appears in
\inl{square} (Listing~\ref{lst:double}, right).
For the nested function \inl{sqr}, the parameter
\inl{n} is unit polymorphic \punitA{} and thus the unit of the result
is \punitA{} squared. Variables \inl{x}, \inl{y}, \inl{t} and \inl{s}
have concrete units which are coherent with instantiations of \inl{sqr}.

The importance of polymorphism as an effort-saving mechanism cannot be
underestimated: without polymorphism the programmer would be forced to
write redundant copies of the function \inl{d} and \inl{sqr} in these examples. More fundamentally,
polymorphism describes essential structure and properties of programs. For example,
\citet{Kennedy:1996} proves that functions without side-effects, such
as \inl{d}, that accept and return the same units can be fully
described in their behaviour by a single constant multiplier. In other
words, for every single-parameter function \inl{d} that accepts and
returns the same units {\tt 'a}, there must exist a number $k$ such that the
function $\lambda n. k\cdot n$ is extensionally equivalent to
\inl{d}. Similarly, for each function like \inl{square} that returns
{\tt 'a}\inl{**2}, there exists some $k$ such that the function is
extensionally equivalent to $\lambda n. k\cdot n^2$. This property
is known as \emph{dimensional invariance} or \emph{invariance under scaling}~\citep{atkey2015models}.

\subsection{Modular Fortran}
Real programs are usually much larger than these toy examples and
contain many modules. The Fortran 90 module system provides separation
of name-spaces, and allows a module to import names from another
module, possibly after renaming. Large Fortran programs are compiled
as separate units and then linked together to form executables. Many
different third-party libraries might be referenced by a given program, but
the source code for those libraries may not be available.

We extended the program analysis engine in CamFort to support multiple
modules and separate analysis of units. This mechanism may also be
used to provide annotations for closed-source libraries. An
illustrative example of inference across modules is given by a
modularised version of the ballistics example in
Listing~\ref{lst:ballmod}, comprising a
\inl{helper} module (left)
imported by the
\inl{use} keyword
in the \inl{ballistics} program (right),
bringing its variables and functions into the
current scope. The main calculation is now performed by the function
\inl{x(t)} which uses the polymorphic \inl{square}.

\begin{listing}[t]
\hspace{0.5em}\begin{minipage}{0.43\linewidth}
\begin{minted}{fortran}
module helper
  implicit none
  != unit :: speed = metre / sec
  != unit metre :: x0
  real, parameter :: x0 = 0

  != unit speed :: v0
  real, parameter :: v0 = 20
  real, parameter :: a = -9.8
contains
  real function square(n)
    real :: n
    square = n * n
  end function square
end module helper
\end{minted}
\end{minipage}
\begin{minipage}{0.47\linewidth}
\begin{minted}{fortran}
! Ballistics program unit

program ballistics
  use helper
  implicit none
  real :: t1 = 2.0, t2 = 10.0
  real :: xsum = x(t1) + x(t2)

contains

  real function x(t)
    real :: t
    x = 0.5 * a * square(t) + v0 * t + x0
  end function x
end program ballistics
\end{minted}
\end{minipage}
\caption{Helper module (left) and ballistics program (right).}
\label{lst:ballmod}\label{lst:helpmod}
\end{listing}

The annotations for variables \inl{x0} and \inl{v0} are given in the
\inl{helper} module. We also show the units alias feature by giving
the name \inl{speed} to \inl{metre / sec}, which acts as a simple
substitution. The function \inl{square} has been
introduced with no unit annotations.

\noindent
Listing~\ref{lst:inferhelper} (left) shows the output of units
inference on the \inl{helper} module. The user-supplied annotations
(for \inl{x0} and \inl{v0}) are output along with two newly inferred
annotations for \inl{square} and its parameter \inl{n}. Since
\inl{square} does not specialise on any particular units, a
freshly-generated polymorphic units variable is inferred for its
parameter \inl{n}, which is squared in its return unit.

\begin{listing}[t]
\vspace{-0.5em}
\hspace{-3em}\begin{minipage}{0.45\linewidth}
\begin{Verbatim}
$ camfort units-infer helper.f90
helper.f90:
  5:22 unit metre :: x0
  5:30 unit metre / sec :: v0
  7:3 unit ('a)**2 :: square
  8:13 unit 'a :: n
\end{Verbatim}
\vspace{5em}
\end{minipage}
\begin{minipage}{0.45\linewidth}
\begin{Verbatim}
$ camfort units-compile helper.f90
Compiling units for 'helper.f90'

$ camfort units-infer ballistics.f90
helper.fsmod: parsed precompiled file.
ballistics.f90:
  5:11 unit sec :: t1
  5:21 unit sec :: t2
  6:11 unit metre :: xsum
  8:3 unit metre :: x
  9:13 unit sec :: t
\end{Verbatim}
\end{minipage}

\vspace{1em}
\caption{Output showing inferred polymorphism (left) and compilation
  of included modules (right).}

\label{lst:inferhelper}\label{lst:inferball}
\vspace{-1em}
\end{listing}

Listing~\ref{lst:inferball} (right) shows the output of two
commands. First, we run a new mode of operation named
\inl{units-compile} on the \inl{helper} module to produce a special
precompiled file named \inl{helper.fsmod}, which contains the
information necessary to import units information into other
modules. These precompiled files are detailed in
Section~\ref{sec:fsmod}. Then we run \inl{units-infer} on the
\inl{ballistics} module, which automatically detects and loads the precompiled
file for \inl{helper}. Although the \inl{ballistics}
module does not contain any annotations the inference is successful
and the output identifies five variables that can be assigned units:
\inl{t1}, \inl{t2}, \inl{xsum}, \inl{x} and \inl{t}. This
result brings together several aspects of our work:
separately-compiled module information (Section~\ref{sec:fsmod}),
interprocedural
knowledge (Section~\ref{sec:interprocedural}), %
intrinsic functions (Section~\ref{sec:intrinsics}),
and polymorphic units (Section~\ref{sec:backendpoly}). The solution
to the units constraints is then computed by a modified Hermite Normal
Form procedure (Section~\ref{sec:modhnf}).

\section{Inference algorithm}
\label{sec:algorithm}

Our core algorithm provides both inference of units in a program and
checking for those variables which have been given a units
specification. The algorithm proceeds by a syntax-directed generation
of constraints (Section~\ref{sec:constraintgen}), an intermediate
stage of polymorphic specialisation (Section~\ref{sec:polyinst})
followed by a solving procedure (Section~\ref{sec:modhnf}). We explain
here also further details related to polymorphic inference
(Section~\ref{sec:backendpoly}) and separate verification
(Section~\ref{sec:fsmod}). The constraints are attached to the
abstract syntax tree and can be traced back later for error-reporting
purposes.

Our implementation applies our algorithm
to Fortran code across versions 66, 77, 90, and 95.
However, since the syntax of Fortran is extensive~\citep{Fortran},
we explain the algorithm on a subset ---
a core imperative language. Figure~\ref{fig:fortransyntax} outlines
the syntax which comprises expressions, contained in statements,
contained within blocks, contained within program units (e.g., functions,
modules).
We use the following syntactic conventions:
\begin{itemize}[leftmargin=1.5em]
\item Variables are denoted by $x$; array subscripting by $\sub x{\cdots}$
\item Functions are denoted by $f$ (including built-in functions) and subroutines
  by $s$;
\item Module and program names are denoted by $m$;
\item Literals are denoted by $n$ or $0$ to which unique identifiers
are associated ranged over by $l$;
\item Unique `call identifiers' for function and subroutine calls are denoted by $i$;
\item The return and parameter variables to functions and subroutines are numbered $0,1,\ldots,k$;
\item Explicitly annotated parametric polymorphic units variables
are denoted by $\alpha$.
\end{itemize}
Additional notation is explained on-demand. All names (variables, function
names) are assumed unique following the scoping rules of the
underlying language.

In Figure~\ref{fig:fortransyntax}, the first rule in the non-terminal
node $B_{\mathit{spec}}$ captures our key
notion of units annotation via comments.
These annotations are intended to
be placed directly before the declarations of the variables they reference. Similarly, we
support annotation of function return units by placing an annotation
just before a function definition, as captured by non-terminal $pu_i$.
Our system also supports aliasing of complex monomorphic units
$u$ to a ${\textit{name}}$, implemented by simple substitution.

\begin{figure}[t]
  \centering
  \begin{align*} %
    e &\bndef\ -e \mid e \oplus e \mid e \oslash e \mid e \inlm{*}
        e\mid e \inlm{/} e \mid e \inlm{**} n \mid n \mid x \mid \sub f{\overline e} \mid \sub x {\overline e}\\[-0.3em]
    &\phantom{\bndef}\text{where}~\oplus\in\set{\inlm{+},\inlm{-}},~\oslash\in\set{\inlm{==},\inlm{/=},\inlm{<},\inlm{<=},\inlm{>},\inlm{>=}},
    \text{and $n$ ranges over numeric literals} \\
    st &\bndef\ x = e \mid x(\overline e) = e \mid \inlm{call}~\sub s{\overline e} \\[-0.3em]
    &\phantom{\bndef}\text{where $s$ includes built-in subroutines that omit the keyword \inl{call}} \\
    B_{\mathit{act}} &\bndef\ st \;\mid\; \inlm{if}~e~\inlm{then}~\overline B_{\mathit{act}}~\inlm{else}~\overline B_{\mathit{act}}~\inlm{end if}
      \;\mid\; \inlm{do while}~\prn{e}~\overline B_{\mathit{act}}~\inlm{end do}\\
    B_{\mathit{spec}} &\bndef\ \inlm{!= unit}~u~\inlm{::}~x \;\mid\;
      \inlm{real ::}~x \; \mid \, \inlm{integer ::}~x
      \; \mid \; \inlm{dimension ::}~\sub x{\overline n} \; \mid \, \inlm{external ::}~x\\
    B &\bndef\ B_{\mathit{spec}} \mid B_{\mathit{act}}\\
    cs &\bndef\ \inlm{contains}~\overline{pu_i}\\
    pu_i &\bndef\ \inlm{!= unit}~u~\inlm{::}~f  \\[-0.4em]
    &\synalt \inlm{!= unit}~\inlm{::}~name~\inlm{=}~u \\[-0.4em]
    &\synalt \inlm{function}~\sub f{\overline x}~\left[\sub{\inlm{result}}{x}\right]~\overline B_{\mathit{spec}}, \overline B_{\mathit{act}}~\left[cs\right]
    \;\mid\; \inlm{subroutine}~\sub s{\overline x}~\overline B_{\mathit{spec}}, \overline B_{\mathit{act}}~\left[cs\right] \\
    pu &\bndef pu_i \;\mid\; \inlm{module}~m~\overline B_{\mathit{spec}}~\left[cs\right]
      \;\mid\; \inlm{program}~m~\overline B_{\mathit{spec}}, \overline
        B_{\mathit{act}}~\left[cs\right]
  \end{align*}
  \vspace{-1.5em}
  \caption{Simplified Fortran syntax}
  \label{fig:fortransyntax}
  \vspace{-0.5em}
\end{figure} 
\subsection{Units Constraints}
\label{sec:constraintgen}

\newcommand{\uvar}{\ensuremath{\mathit{uvar}}}
\begin{definition}[Units syntax and representation]
  The internal representation of units types (i.e., not the surface
  syntax of Figure~\ref{fig:syntax}) is ranged over by $u$:
\begin{align*}
u & \; \bndef \; U \,\mid\, \unitless \,\mid\, u_1 \, u_2 \,\mid\, u^{z} \,\mid\,
    \uvar
\end{align*}
where $z \in \mathbb{Z}$ and
$U$ is any concrete unit, e.g., {\tt metre}, {\tt sec}
and $\unitless$ is the `unitless' type. In our
implementation, the concrete syntax of units variables \uvar{} is in the ML style of
type variables: e.g., {\tt 'a}. Internally, \uvar{} splits
into a more fine-grained definition:
\begin{align*}
  \mathit{uvar} & \; \bndef \; \Uvar {\mathit{lx}} \,\mid\, \Uabs {\mathit{fs}} {\mathit{klx}} \,\mid\, \Uuse {\mathit{fs}} {\mathit{klx}} i \,\mid\, \UEAPabs{\alpha} \,\mid\, \UEAPuse{\alpha}{i}
\end{align*}
where
 $\mathit{lx}$ ranges over unique identifiers for literals $l$ or variables $x$;
 $\mathit{fs}$ ranges over function and subroutine names;
 $\mathit{klx}$ ranges over $\mathit{lx}$ or function parameters $k$;
 and $\alpha$ ranges over explicit polymorphic-units variables.
 Therefore, $\Uvar {\mathit{lx}}$ is a units variable for literals or variables.
 $\Uabs {\mathit{fs}} {\mathit{klx}}$ is an abstract (universally quantified) units variable for $\mathit{klx}$ inside function/subroutine $\mathit{fs}$,
 $\UEAPabs {\alpha}$ is an abstract units variable for polymorphic units-variable $\alpha$ coming from an explicit signature,
 $\Uuse {\mathit{fs}} {\mathit{klx}} i$ is a units variable for $\mathit{klx}$ in function or subroutine $\mathit{fs}$ instantiated at unique call site $i$, and
 $\UEAPuse {\alpha} i$ is a units variable for $\alpha$ instantiated at unique call site $i$.
\end{definition}

\begin{figure*}[b]\footnotesize
\begin{minipage}{\textwidth}
\centering
\begin{equation*}
\begin{gathered}
  \erule{polyvar}{\ujudge{\Gamma}{x}{\Uabs f x}{\cdot}}{\polycontext{f,\Gamma}&x\in\Gamma}
  \qquad\erule{var}{\ujudge{\Gamma}{x}{\Uvar x}{\cdot}}{\neg\polycontext{f,\Gamma}&x\in\Gamma}
  \qquad\erule{\oplus}{
    \ujudge{\Gamma}{e_1\oplus e_2}{u_1}{\Delta_1,\Delta_2,u_1\coneq u_2}
  }{
    \ujudge{\Gamma}{e_1}{u_1}{\Delta_1} & \ujudge{\Gamma}{e_2}{u_2}{\Delta_2}
  }
  \\[0.3em]
  \erule{\oslash}{
    \ujudge{\Gamma}{e_1\oslash e_2}{\unitless}{\Delta_1,\Delta_2,u_1\coneq u_2}
  }{
    \ujudge{\Gamma}{e_1}{u_1}{\Delta_1} & \ujudge{\Gamma}{e_2}{u_2}{\Delta_2}
  }
  \qquad\erule{mul}{\ujudge{\Gamma}{e_1\inlms{*}e_2}{u_1u_2}{\Delta_1,\Delta_2}}{\ujudge{\Gamma}{e_1}{u_1}{\Delta_1} & \ujudge{\Gamma}{e_2}{u_2}{\Delta_2}}
  \qquad\erule{div}{\ujudge{\Gamma}{e_1\inlms{/}e_2}{u_1u_2^{-1}}{\Delta_1,\Delta_2}}{\ujudge{\Gamma}{e_1}{u_1}{\Delta_1} & \ujudge{\Gamma}{e_2}{u_2}{\Delta_2}}
  \\[0.3em]
  \erule{unary}{\ujudge{\Gamma}{{- e}}u{\Delta}}{\ujudge{\Gamma}{{e}}u{\Delta}}
  \qquad\erule{power}{
    \ujudge{\Gamma}{e\inlms{**}n}{u^n}{\Delta}
  }{
    \ujudge{\Gamma}{e}{u}{\Delta} \quad n \in \mathbb{Z}
  }
  \qquad\erule{subscript}{
    \ujudge{\Gamma}{\sub x{\overline e}}u{\Delta}
  }{
    \ujudge{\Gamma} x u {\cdot} & \isarray{x,\Gamma} & \ujudge{\Gamma}{\overline e}{\overline {\unitless}}{\Delta}
  }
  \\[0.3em]
  \erule{fcall}{
    \ujudge{\Gamma}{\sub f {\overline e}}{\Uuse f 0 i}{
      \Delta,u_1\coneq\Uuse f 1 i,\ldots,u_k\coneq\Uuse f k i
    }
  }{
    f:\external\not\in\Gamma & \freshcallid i & \isfunction{f,\Gamma} &
    \ujudge{\Gamma}{\overline e}{\overline u}{\Delta} & \text{where}~\overline u=u_1,\ldots,u_k
  }
  \\[0.3em]
  \erule{fcallext}{
    \ujudge{\Gamma}{\sub f{\overline e}}{\Uuse f 0 0}{
      \Delta,u_1\coneq\Uuse f 1 0,\ldots,u_k\coneq\Uuse f k 0
    }
  }{
    f:\external\in\Gamma & \isfunction{f,\Gamma} &
    \ujudge{\Gamma}{\overline e}{\overline u}{\Delta} & \text{where}~\overline u=u_1,\ldots,u_k
  }
  \\[0.3em]
  \erule{list}{
    \ujudge{\Gamma}{\overline e}{\overline u}{\Delta}\qquad\text{where}~\overline e=e_1,e_2,\ldots,e_k~\text{and}~\overline u=u_1,u_2,\ldots,u_k
  }{
    \ujudge{\Gamma}{e_1}{u_1}{\Delta_1} & \ujudge{\Gamma}{e_2}{u_2}{\Delta_2} & \cdots & \ujudge{\Gamma}{e_k}{u_k}{\Delta_k} & \Delta=\Delta_1,\Delta_2,\ldots,\Delta_k
  }
\end{gathered}
\end{equation*}
\end{minipage}
\vspace{-0.5em}
  \caption{Rules for expressions (apart from literals, see Figure~\ref{fig:literalrules})}
  \label{fig:expressionrules}
\end{figure*}

\subsubsection{Expressions}\label{sec:expressions}

Our inference algorithm assigns units to expressions in
a compositional manner, generating a set of units equality
constraints. Figure~\ref{fig:expressionrules} gives the rules for
expressions defined by judgments of the form
$\ujudge{\Gamma}{e}{u}{\Delta}$
 meaning $e$ has units $u$ in context $\Gamma$ and generates a set
of units equality constraints $\Delta$, whose elements are written
$u_1 \coneq u_2$. Variable environments $\Gamma$ are ordered sequences
mapping variables to type information (\inlm{real}, \inlm{integer},
$\Array$), their unit, and an optional attribute
$\external$ to denote any externally defined (thus abstract)
functions and variables.

The units of variables and literals depend on
whether they are at the top-level or inside a function or subroutine, where they may be
polymorphic. Such contexts are classified by a predicate:
\begin{definition}\label{def:polycontext}
  Given a function or subroutine name $\mathit{fs}$, then $\Gamma$ is a \emph{polymorphic context}
  for $\mathit{fs}$ if:
  \begin{align*}
    \polycontext{\mathit{fs},\Gamma} := \quad  \exists \Gamma_1, \Gamma_2, x \, . \; \Gamma &\equiv (\Gamma_1, x:\Uabs {\mathit{fs}} x, \Gamma_2) \;\wedge\\
      \forall {x'}, \mathit{fs}'& \, . \; x' : \Uabs {\mathit{fs}'} {x'} \in \Gamma_2 \implies \mathit{fs}=\mathit{fs}'
  \end{align*}
  That is, $\mathit{fs}$ is the most recently enclosing function or
  subroutine in the context if there is a parameter $x$ in the
  environment associated with $\mathit{fs}$, and there is no other
  function or subroutine $\mathit{fs}'$ ahead of it.
\end{definition}
\noindent
For example, {polycontext} is used in \erulenamed{polyvar} (Figure~\ref{fig:expressionrules})
to generate units for a variable $x$
in a function $f$ as the abstract units variable $\Uabs f x$. By
contrast, variables in a monomorphic setting are assigned $\Uvar x$,
which cannot later be instantiated.
The rest of the expression rules generate various constraints:
compound units for multiplication and division (\erulenamed{mul},
\erulenamed{div}), equality constraints for addition
(\erulenamed{\oplus}) and boolean comparators (\erulenamed{\oslash}),
preservation of units for negation (\erulenamed{unary}), and
exponentiation of units by integer constants (\erulenamed{power}).
Function calls are by \erulenamed{fcall} and \erulenamed{fcallext}.

Figure~\ref{fig:expressionrules} uses some further helper predicates on
environments and for fresh identifier generation:
\begin{itemize}[leftmargin=1em]
\item $\freshcallid i$ asserts that $i$ is a fresh identifier in
this derivation (any another instance of $\freshcallid j$ is
guaranteed to satisfy $i\neq j$);
\item $\isfunction{f,\Gamma}$ and $\isarray{x,\Gamma}$ classify
$f$ and $x$ based on their (normal) types in $\Gamma$.%
\end{itemize}

\begin{figure}[t]
  {\footnotesize{
  \begin{equation*}
  \erule{polyzero}{\ujudge{\Gamma}{0}{\Uabs f
      l}{\cdot}}{\polycontext{f,\Gamma} & \freshlitid l} \;\;
  \erule{polylit}{\ujudge{\Gamma}{n}{\unitless}{\cdot}}{\polycontext{f,\Gamma}
    & n\neq 0} \;\;
  \erule{literal}{\ujudge{\Gamma}{n}{\Ulit l}{\cdot}}{\neg\polycontext{f,\Gamma} & \freshlitid l}
\end{equation*}}}
\vspace{-1em}
\caption{Rules for literal expressions}
\label{fig:literalrules}
\end{figure}

\subsubsection{Literals}\label{sec:literals}
The rules for literals bear some
explanation and are in the separate
Figure~\ref{fig:literalrules}. There are three relevant
rules that encode the handling of literal values recommended by
\citet{Kennedy:1996}. \erulenamed{polyzero} provides a `polymorphic
zero' when used in a context where polymorphism is
allowed. \erulenamed{polylit} requires all other literal values found
in a polymorphic context to be unitless. Outside of a polymorphic
context, \erulenamed{literal} generates a fresh units variable to be
solved; this differs to Kennedy's approach which we explain in more
detail in Section~\ref{sec:fsharp}.

Since zero is the additive
identity, and addition requires that all terms have the same units,
the only literal value that can act polymorphically is zero. If any
other literal value could have polymorphic units then a simple scalar
multiplication would unify with any units, rendering the checking
procedure useless. For example, if nonzero polymorphic literals were
allowed and we modified the code in Listing~\ref{lst:ball1} to be %
\inl{0.5 * a * t + v0 * t}, %
then it would pass the checker because it would assume that the
literal value \inl{0.5} could be given units \inl{sec}.

Still, we need a way for literals to be given units within functions
and subroutines. In the \fsharp\ units-of-measure extension you must
annotate literals directly, for example:%
\begin{minted}{fsharp}
let w = 2.0<kg> + 3.5<kg>
\end{minted}
However we avoid direct annotation because it violates our incremental
principle of \emph{harmlessness}\label{sec:harmlessness}: to maintain usability of existing
compilers and tools. We could instead create inline comment-based
annotations specifying the units, but this is cumbersome, not
supported by all Fortran compilers and potentially confusing. We
found that the most practical way of assigning units to nonzero
literals is to carve out an exception for the case of assigning a
literal value to a monomorphically annotated variable. In that case,
no constraint is generated. This condition is expressed by the rule
\strulenamed{assignliteral} in Figure~\ref{fig:statementrules} where
$\isMonomorphic u$ is true iff $u$ does not have any parametric
polymorphic units variables within it. When \inl{x} is a Fortran
variable annotated by monomorphic units, simple assignment statements
such as \inl{x = 2.0} are allowed, but any compound expression
involving nonzero literals falls back on \erulenamed{polylit}. This
particular design choice has the virtue of catching the common case of
declaring constants %
(e.g. \inl{real :: x = 2.0}) %
as well as discouraging the poor practice
\citep{Martin:2009:CleanCode} of scattering `magic numbers' throughout
code (e.g. \inl{x = y - 9.381}).

\subsubsection{Casting}\label{sec:casting}
The literal exception allows us to create `casting' values that help
convert between concrete units. For example, you can declare %
\inl{!= unit cm/in :: inchesToCm} followed by %
\inl{real :: inchesToCm = 2.54}. %
Abuse of casting can allow you to sidestep the units system but it is
a practical feature that cannot be avoided in real work, and this way
the conversion factors are named and documented rather than being
scattered around the code. We would argue that better practice would
be to isolate all of the necessary conversion factors in a common
module shared by the entire program (or better yet, a library we may
distribute along with our work). Conversion factor literals can only
be given monomorphic units. Section~\ref{sec:intrinsics} discusses
explicit, direct unit casts.

\subsubsection{Monomorphism restriction for nonzero literals}\label{sec:monorestrict}
The restriction of the literal exception to monomorphic units is
important for soundness. Otherwise it would be possible to write
programs that circumvent the relationships expressed through
polymorphism. For example, a function with units-polymorphic variable
\inl{x} could simply run a loop that computes $x^2=\sum_1^{x}x$ (see Listing~\ref{lst:cpfunits}). The
units inference process would simply interpret that as a series of
additions that return the exact same units as provided, rather than
the square of those units. More generally, allowing the user to
annotate nonzero literals with polymorphic units would break the
theorem of dimensional invariance~\citep{Kennedy:1996,atkey2015models}, which states
that the behaviour of dimensionally-correct equations is independent
of the specific units-of-measure used. For a counterexample, suppose
that we were actually allowed to write a units-polymorphic
function with a polymorphic literal, like $\sub f x=x+2$. Then it would not satisfy its dimensional
invariance theorem that $\sub f {k\cdot x}=k\cdot \sub f {x}$ for any $k$. To put some numbers on it, if we
take $x=1~\text{inch}$ then we find that $\sub f x=3~\text{inches}$. However with
$k=2.54~\text{cm/inch}$ then $\sub f {2.54\cdot x}=4.54~\text{cm}$, which is not
equal to $2.54\cdot \sub f x=7.62~\text{cm}$.

\subsubsection{Function calls}
There are two rules for function calls, \erulenamed{fcall} and
\erulenamed{fcallext}. The reason is because Fortran supports passing
functions as parameters, but does not have the concept of `function
types' to properly describe their behaviour. Instead, such function
parameters are simply declared \inl{external}. We support polymorphic
functions but not arbitrary rank
polymorphism~\citep{Jones:2007:ArbitraryRank}. That would allow a
function parameter to be instantiated with different units at each
call-site, but would easily generate inconsistencies under the rules
of our system. To avoid this problem we collapse all external function
parameter calls to call id 0 (chosen arbitrarily). That permits only a
single polymorphic instance to be created, which is shared by calls to
that function parameter within the enclosing function body. Future
work is to allow explicit annotations on Fortran \inl{interface} blocks
to allow arbitrary rank polymorphism.

\begin{figure*}[b]
{\small{\begin{minipage}{\textwidth}
\centering
\begin{equation*}
\begin{gathered}
  \strule{call}{
    \ststep{\Gamma}{\Delta}{\inlms{call}~\sub s{\overline e}}{\Delta,\Delta',u_1\coneq\Uuse s 1 i,\ldots,u_k\coneq\Uuse s k i}
  }{
    \freshcallid i & \ujudge{\Gamma}{\overline e}{\overline u}{\Delta'} &\text{where}~\overline u=u_1,u_2,\ldots,u_k
  }\\[0.4em]
  \strule{assignliteral}{
    \ststep{\Gamma}{\Delta}{x=n}{\Delta}
  }{
    n\neq 0 & \ujudge{\Gamma}{x}{\Uvar x}{\cdot} & \forall u. \prn{\Uvar x\coneq u \in \Delta \implies \isMonomorphic u}
  }\\[0.4em]
   \strule{assignzero}{
    \ststep{\Gamma}{\Delta}{x=0}{\Delta}
  }{
  }\qquad\qquad
  \strule{assign}{
    \ststep{\Gamma}{\Delta}{x=e}{\Delta,\Delta_2,u_1\coneq u_2}
  }{
    \ujudge{\Gamma}{x}{u_1}{\cdot} & \ujudge{\Gamma}{e}{u_2}{\Delta_2} & e\neq n
  }\\[0.4em]
  \strule{arrayassign}{
    \ststep{\Gamma}{\Delta}{\sub x{\overline e}=e}{\Delta,\Delta_2,\Delta_3,u_1\coneq u_2}
  }{
    \ujudge{\Gamma}{x}{u_1}{\cdot} & \ujudge{\Gamma}{e}{u_2}{\Delta_2}  & (e\neq n \vee e=0) & \ujudge{\Gamma}{\overline e}{\overline{\unitless}}{\Delta_3}
  }
\end{gathered}
\end{equation*}
\end{minipage}}}
\caption{Rules for statements}
\label{fig:statementrules}
\vspace{-0.5em}
\end{figure*}

\begin{figure*}[t]\footnotesize
\begin{minipage}{\textwidth}
\centering
\begin{equation*}
\begin{gathered}
  \blrule{list}{\cstep{\Gamma_0}{\Delta_0}{\overline B}{\Gamma_k}{\Delta_k}\qquad \text{where}~\overline B = B_1,B_2,\ldots,B_k}{\cstep{\Gamma_0}{\Delta_0}{B_1}{\Gamma_1}{\Delta_1} & \cstep{\Gamma_1}{\Delta_1}{B_2}{\Gamma_2}{\Delta_2} & \cdots & \cstep{\Gamma_{k-1}}{\Delta_{k-1}}{B_k}{\Gamma_k}{\Delta_k}} \\
  \blrule{annotpoly}{
    \cstep{\Gamma}{\Delta}{\inlms{!= unit U :: x}}{\Gamma}{\Delta,\Uabs f x\coneq \inlms{U}}
  }{
    \polycontext{f,\Gamma}
  }\;\;
  \blrule{annotmono}{
    \cstep{\Gamma}{\Delta}{\inlms{!= unit U :: x}}{\Gamma}{\Delta,\Uvar x\coneq \inlms{U}}
  }{
    \neg\polycontext{f,\Gamma}
  }\\
  \blrule{statement}{\cstep{\Gamma}{\Delta}{st}{\Gamma}{\Delta'}
  }{
    \ststep{\Gamma}{\Delta}{st}{\Delta'}
  }\qquad
  \blrule{typedecl}{
    \cstep{\Gamma}{\Delta}{type~\inlms{::}~x}{\Gamma,x:type}{\Delta}\qquad\text{where}~type\in\set{\inlms{real},\inlms{integer}}
  }{}\\
  \blrule{dimdecl}{
    \cstep{\Gamma}{\Delta}{\inlms{dimension ::}~\sub x{\overline n}}{\Gamma,x:\Array}{\Delta}
  }{}\qquad
  \blrule{extdecl}{
    \cstep{\Gamma}{\Delta}{\inlms{external ::}~x}{\Gamma,x:\external}{\Delta}
  }{}\\
  \blrule{if}{
    \cstep{\Gamma_0}{\Delta_0}{\inlms{if}~e~\inlms{then}~\overline B_1~\inlms{else}~\overline B_2~\inlms{end if}}{\Gamma_0}{\Delta_2,\Delta_e}
  }{
    \ujudge{\Gamma_0}{e}{\unitless}{\Delta_e} & \cstep{\Gamma_0}{\Delta_0}{\overline B_1}{\Gamma_1}{\Delta_1} & \cstep{\Gamma_0}{\Delta_1}{\overline B_2}{\Gamma_2}{\Delta_2}
  }\;\;
  \blrule{while}{
    \cstep{\Gamma_0}{\Delta_0}{\inlms{do while}~\prn{e}~\overline B~\inlms{end do}}{\Gamma_0}{\Delta_1,\Delta_2}
  }{
    \ujudge{\Gamma_0}{e}{u}{\Delta_1} & \cstep{\Gamma_0}{\Delta_0}{\overline B}{\Gamma_1}{\Delta_2}
  }
\end{gathered}
\end{equation*}
\end{minipage}
\caption{Rules for blocks}
\label{fig:blockrules}
\end{figure*}

\subsubsection{Statements and Blocks}

Statements (in the usual imperative sense) are typed by judgments
that have an incoming context $\Gamma$, set of constraints
$\Delta$ and an outgoing constraint set $\Delta'$:
\begin{equation*}
\ststep{\Gamma}{\Delta}{st}{\Delta'}
\end{equation*}
The rules are given in Figure~\ref{fig:statementrules}.

Blocks include syntax such as conditionals or loops, which have a
similar form to statements and are typed by judgments of the form:
\begin{equation*}
\cstep{\Gamma}{\Delta}{B}{\Gamma'}{\Delta'}
\end{equation*}
is a relation between $\prn{\Gamma,\Delta,B}$ and the subsequently
produced context and constraint set $\prn{\Gamma',\Delta'}$ (Figure~\ref{fig:blockrules}). Blocks
can include variable-type declarations (\blrulenamed{dimdecl},
\blrulenamed{extdecl}, \blrulenamed{typedecl}), and thus they generate a
new typing context $\Gamma'$.

The rules for \blrulenamed{annotpoly} and \blrulenamed{annotmono} in
Figure~\ref{fig:blockrules} are key to connecting the generated units constraints with
annotations of concrete units. The \blrulenamed{annotmono} rule introduces concrete units
\inl{U} for variable $x$ and thus generates a constraint $\Uvar x
\coneq \texttt{U}$; \blrulenamed{annotpoly} is similar but in a
polymorphic context hence $x$ may be units-polymorphic
and so the generated constraint applies the units variable representation
$\Uabs f x \coneq \texttt{U}$.

Rules for conditionals (\blrulenamed{if}) and loops
(\blrulenamed{while}) recursively apply
expression typing and block-typing rules, composing generated
constraint and context sets sequentially. Thus, we perform a
`may analysis': we take the conjunction of
constraints from all control-flow paths, similar to a standard
type inference.

\subsubsection{Program Units}

We introduce the notion of a template $T$ which maps
function and subroutine names to their associated constraints, in order to
support polymorphism (Section~\ref{sec:polyinst}) and to be
used in precompiled (preverified) files (Section~\ref{sec:fsmod}). Program units are typed
in Figure~\ref{fig:programunitrules} by judgments of the following form:
\begin{equation*}
\pustep{\Gamma}{T}{pu}{T'}
\end{equation*}
is a relation between $\prn{\Gamma,T,pu}$ and the subsequently
produced template $T'$.

Program units are grouped into files which can be modules
or top-level programs. Module-maps $M$
map module names to tuples $\prn{\Delta,T}$ where $\Delta$
contains constraints associated with module variables and $T$ are
templates. A final judgment form coalesces information from multiple modules:
\begin{equation*}
\modstep{M}{pu}{M'}
\end{equation*}
is a relation between $\prn{M,pu}$ and the subsequently produced
module-map $M'$.

\begin{figure*}[t]\small
\begin{minipage}{\textwidth}
\centering
\begin{equation*}
\begin{gathered}
  \purule{funres}{
    \pustep{\Gamma}{T}{
      \inlms{function}~\sub f{x_1,\ldots,x_k}\sub{\inlms{result}}{x_0}~\overline B_{\mathit{spec}}, \overline B_{\mathit{act}}~\inlms{contains}~\overline{pu}
    }{
      T',\prn{f\mapsto\Delta',\Delta''}
    }
  }{
    \deduce{
      \cstep{\Gamma_1}{\Delta}{\overline{B}_{\mathit{act}}}{\Gamma'}{\Delta'} \qquad
      \text{let}~\Delta''=\Gamma'\!\prn{x_0}\coneq\Uabs f 0,\ldots,\Gamma'\!\prn{x_k}\coneq\Uabs f k
    }{
      \cstep{\cdot}{\cdot}{\overline{B}_{\mathit{spec}}}{\Gamma_0}{\Delta} &\;
      \text{let}~\Gamma_1={\Gamma_0,x_0:\Uabs f {x_0},\ldots,x_k:\Uabs f {x_k}}\; &
      \pustep{\Gamma_1}{T}{\overline{pu}}{T'}
    }
  }\\[0.4em]
  \purule{fun}{
    \pustep{\Gamma}{T}{
      \inlms{function}~\sub f{x_1,\ldots,x_k}~\overline B_{\mathit{spec}}, \overline B_{\mathit{act}}~\inlms{contains}~\overline{pu}
    }{T'}
  }{
    \pustep{\Gamma}{T}{
      \inlms{function}~\sub f{x_1,\ldots,x_k}\sub{\inlms{result}}{f}~\overline B_{\mathit{spec}}, \overline B_{\mathit{act}}~\inlms{contains}~\overline{pu}
    }{T'}
  }\\[0.4em]
  \purule{subr}{
    \qquad
    \pustep{\Gamma}{T}{
      \inlms{subroutine}~\sub s{x_1,\ldots,x_k}~\overline B_{\mathit{spec}}, \overline B_{\mathit{act}}~\inlms{contains}~\overline{pu}
    }{T',\prn{s\mapsto\Delta',\Delta''}}\qquad
  }{
    \deduce{
      \cstep{\Gamma_1}{\Delta}{\overline{B}_{\mathit{act}}}{\Gamma'}{\Delta'}\qquad
      \text{let}~\Delta''=\Gamma'\!\prn{x_1}\coneq\Uabs s 1,\ldots,\Gamma'\!\prn{x_k}\coneq\Uabs s k
    }{
      \cstep{\cdot}{\cdot}{\overline{B}_{\mathit{spec}}}{\Gamma_0}{\Delta} &\;
      \text{let}~\Gamma_1={\Gamma_0,x_1:\Uabs s {x_1},\ldots,x_k:\Uabs s {x_k}}\; &
      \pustep{\Gamma_1}{T}{\overline{pu}}{T'}
    }
  }\\[0.4em]
  \purule{\!\!list}{
    \pustep{\Gamma}{T_0}{pu_1,\ldots,pu_k}{T_k}
  }{
    \pustep{\Gamma}{T_0}{pu_1}{T_1} & \cdots & \pustep{\Gamma}{T_{k-1}}{pu_k}{T_k}
  }\;
  \purule{\!\!module}{
    \modstep{M}{\inlms{module}~m~\overline B_{\mathit{spec}}~\inlms{contains}~\overline{pu}}{M,m\mapsto\prn{\Delta,T}}
  }{
    \cstep{\cdot}{\cdot}{\overline B_{\mathit{spec}}}{\Gamma}{\Delta} &
    \pustep{\Gamma}{\cdot}{\overline{pu}}{T}
  }\\[0.4em]
  \purule{program}{
    \modstep{M}{\inlms{program}~m~\overline B_{\mathit{spec}},\overline B_{\mathit{act}}~\inlms{contains}~\overline{pu}}{M,m\mapsto\prn{\Delta',T'}}
  }{
    \cstep{\cdot}{\cdot}{\overline B_{\mathit{spec}}}{\Gamma}{\Delta} &
    \pustep{\Gamma}{\cdot}{\overline{pu}}{T} &
    \cstep{\Gamma}{\Delta}{\overline B_{\mathit{act}}}{\Gamma'}{\Delta'}
  }
\end{gathered}
\end{equation*}
\end{minipage}
\caption{Rules for program units}
\label{fig:programunitrules}
\end{figure*}

\subsection{Interprocedural Polymorphic Instantiation}\label{sec:polyinst}\label{sec:interprocedural}

Once the constraints within all program units have been gathered and
stored in a template table we begin the process of instantiating
polymorphic function or subroutine templates as concrete
constraints. We identify all of the uses of a polymorphic function or
subroutine $\mathit{fs}$ by finding pairs $\prn{\mathit{fs},i}$ where $i$ represents the
call-ids associated with function or subroutine calls contained within
the constraints. We call these pairs `instances' and for each instance
we invoke a recursive function \texttt{substInstance} to look up the
associated function or subroutine template and invoke
$\instantiate u i$ (Definition~\ref{def:instantiate}) on all of the
abstract units $u$ within it.

\begin{definition}[Instantiation]\label{def:instantiate}
For units $u$ and a call-site identifier $i$ then
$\instantiate u i$ is defined:
  \begin{equation*}
    \instantiate u i =
    \begin{cases}
      \Uuse {\mathit{fs}} {\mathit{kxl}} i &\text{if }u=\Uabs {\mathit{fs}} {\mathit{kxl}} \\
      \UEAPuse {\alpha} i &\text{if }u=\UEAPabs {\alpha} \\
      u &\text{otherwise}
    \end{cases}
  \end{equation*}
\end{definition}
We distinguish ourselves from Osprey~\citep{Jiang:2006:ICSE} by
supporting units-polymorphic functions at all levels of the
call-graph, not only the leaves. This is accomplished with a call-id
remap function that works on the given template prior to instantiation
and rewrites all uses of polymorphism within the template to have a
fresh call-id specific to this particular call. The
\texttt{substInstance} function recurses down the call-graph
performing both call-id remap and instantiation of abstract units. We
have chosen to stop the recursion when a cycle in the call-graph is
detected. This means we support recursive units-polymorphic functions
in Fortran programs but more complicated `units-polymorphic
recursion'~\citep{Mycroft:1984:PolyRecursion} is not supported: with
recursive calls, polymorphic units are instantiated at the first level
and must remain the same throughout the recursion, trading a small
amount of expressiveness for better tractability. This is a common
choice with type systems and \citet{Kennedy:1996} also made the same
decision in his work.

\subsection{Reduction to Linear Diophantine Equations}\label{sec:modhnf}

We adopt and customise the method discussed by \citet{Wand:1991} and
further by \citet{Jiang:2006:ICSE}. First we extract constraints by
analysing the structure of program code, as described by
Section~\ref{sec:constraintgen}, and then we convert the constraints
to a set of linear Diophantine equations relating units variables that can be solved
simultaneously. Figure~\ref{fig:conversionrules} shows the rules that
translate constraints and units into linear equations with terms
corresponding to both unknown units variables $\unknownv{\cdots}$ and
known base units $\baseunit{\cdots}$. This conversion to linear
equation, and the subsequent solving procedure, deals with the
equational theory of units-of-measure (which form
an Abelian group).

\begin{figure}[t]
  \centering
  \begin{align*}
    & &\conv{u_1\coneq u_2} &\mapsto \conv{u_1} = \conv{u_2} \\
    \conv{u_1 u_2} &\mapsto \conv{u_1} + \conv{u_2} & \conv{u_1^z} &\mapsto z\conv{u_1} & \conv{\unitless} &\mapsto 0 \\
    \conv{U} &\mapsto \baseunit{U} & \conv{\UEAPabs{\alpha}}&\mapsto\baseunit{\alpha} & \conv{\Uvar{fs,\cdots}} &\mapsto \unknownv{fs,\cdots}
  \end{align*}
\vspace{-1em}
  \caption{Conversion of constraints to linear equations, where
    $\conv{-}$ is overloaded on constraints and units}
  \label{fig:conversionrules}
\end{figure}

\begin{figure*}[t]\small
  \begin{minipage}{\textwidth}
  \centering
  \begin{equation*}
\setlength{\arraycolsep}{0.35em}
  \begin{gathered}
    \text{given constraints of the form:}~c_1\lhs_1 + \cdots + c_m\lhs_m = b_1\rhs_1 + \cdots + b_n\rhs_n\\
    \text{where each unknown units variable $\unknownv{\cdots}$ has been assigned a column $C_j$ for $0<j\leq m$}\\
    \text{and each base unit $\baseunit{\cdots}$ has been assigned a column $B_j$ for $0<j\leq n$}
    \\[-0.25em]
    A_{\textit{aug}}=~
    \kbordermatrix{~                      & \lhs_1 & \cdots  & \lhs_m & \vr \rhs_1 & \cdots & \rhs_n \cr
                    \text{constraint}_1  & c_1         & \cdots  & c_m         & \VR b_1         & \cdots & b_n \cr
                    \qquad \vdots        & \vdots      & \ddots  & \vdots      & \VR \vdots      & \ddots & \vdots \cr
    }\qquad
    H=
    \kbordermatrix{~                      & \lhs_1 & \cdots  & \lhs_m & \vr \rhs_1 & \cdots & \rhs_n \cr
                    \text{constraint}_1  & c'_1=1         & \cdots  & c'_m         & \VR b'_1         & \cdots & b'_n \cr
                    \qquad \vdots        & \overset{0}{\vdots}      & \ddots  & \vdots      & \VR \vdots      & \ddots & \vdots \cr
    }
  \end{gathered}
  \end{equation*}
\vspace{-1em}
  \end{minipage}
  \caption{Matrices}
  \label{fig:matrices}
\end{figure*}

Once all constraints have been translated into linear equations, they
are set up into the standard $AX=B$ matrix equation form. For the
purposes of Gaussian elimination-style solvers this is represented as
an augmented matrix $A_{\textit{aug}}=\left[C B\right]$ as shown in
Figure~\ref{fig:matrices}. Each of the unknown terms
$\unknownv{\cdots}$ corresponds to a column in the matrix $C$ on the
left-hand side of $A_{\textit{aug}}$. Each of the known base units
$\baseunit{\cdots}$ corresponds to a column in the matrix $B$ on the
right-hand side of $A_{\textit{aug}}$. The linear equations are transcribed
into rows of $A_{\textit{aug}}$ by taking the co-efficients of each term and
placing them into the appropriate matrix under the corresponding
column for that term.

The Rouch\'e--Capelli theorem~\citep{Shafarevich:2012} states that if
the rank of the coefficient matrix $A$ is not equal to the rank of the
augmented matrix $A_{\textit{aug}}$ then the system of linear equations is
inconsistent, meaning a units error is detected and reported to the user.
Otherwise, the standard approach for solving the equations encoded in $A_{\textit{aug}}$
obtains the unique Hermite Normal Form (HNF) $H$, an upper-triangular
matrix from which the answers can be easily extracted by
back-substitution~\citep{Cohen:2013}.
\begin{definition}[Hermite Normal Form]\label{def:hnf}
  Given a matrix $M$, the Hermite Normal Form (HNF) $H$ is an integer
  matrix with the following properties:
  \begin{itemize}[leftmargin=1.5em]
  \item The leading co-efficient in a row, also known as the {\em pivot}, is
    the first nonzero positive integer in the row.
  \item If $j$ is the column of the pivot in row $i$ and $j'$ is the
    column of the pivot in row $i+1$ then $j < j'$.
  \item If there is a pivot at row $i$ and column $j$ then, for all
    rows $i'$ in the matrix:
    \begin{itemize}
    \item if $i' < i$ then $0\leq \sub H{i', j}<\sub H{i,j}$
    \item if $i' > i$ then $\sub H{i', j} = 0$
    \end{itemize}
  \end{itemize}
There are many well-explored algorithms for finding the HNF and it is
a unique form that can be computed for any integer
matrix~\citep{Cohen:2013}.
\end{definition}

\newcommand\Leadsto[1]{\overset{\mathsf{#1}}{\Rightarrow}\!\!\!}
\newcommand\hi[1]{{\color{cyan}{#1}}}

\begin{figure*}[t]
  \begin{minipage}{\textwidth}
  \centering
  \begin{equation*}
\setlength{\arraycolsep}{0.27em}
  \begin{gathered}
    constraints = \set{\inlm{v}\coneq\inlm{y}, \qquad \inlm{w}\coneq\inlm{y**4}, \qquad \inlm{y**4}\coneq\inlm{x**6}} \\[-0.1em]
    \cdots\Leadsto{HNF}
    \kbordermatrix{~ &  v &  w &  x &  y & \vr \cr
                     &  1 &  0 &  0 & -1 & \VR \cr
                     &  0 &  1 &  0 & -4 & \VR \cr
                     &  0 &  0 &  6 & -4 & \VR \cr
    }
    ~\Leadsto{mod}
    \kbordermatrix{~ &  v &  w &  x &  y & \vr \hi{\alpha} \cr
                     &  1 &  0 &  0 & -1 & \VR  \hi 0 \cr
                     &  0 &  1 &  0 & -4 & \VR  \hi 0 \cr
                     &  0 &  0 &  6 & -4 & \VR  \hi 0 \cr
                     &  \hi 0 & \hi 0 &  \hi 1 &  \hi 0 & \VR\hi{-2} \cr
    }
    ~\Leadsto{HNF}
    \kbordermatrix{~ &  v &  w &  x &  y & \vr \alpha \cr
                     &  1 &  0 &  0 &  \hi 3 & \VR -12 \cr
                     &  0 &  1 &  0 &  0     & \VR -12 \cr
                     &  0 &  0 &  1 &  0     & \VR  -2 \cr
                     &  0 &  0 &  0 &  \hi 4 & \VR -12 \cr
    }
    ~\!\Leadsto{{\mathbb{Z}}~RREF}\!
    \kbordermatrix{~ &  v &  w &  x &  y & \vr \alpha \cr
                     &  1 &  0 &  0 &  0 & \VR \hi{-3} \cr
                     &  0 &  1 &  0 &  0 & \VR \hi{-12} \cr
                     &  0 &  0 &  1 &  0 & \VR \hi{-2} \cr
                     &  0 &  0 &  0 &  1 & \VR \hi{-3} \cr
    }
  \end{gathered}
  \end{equation*}
  \end{minipage}
  \vspace{-0.4em}
  \caption{Example operation of modified Hermite Normal Form algorithm}
  \label{fig:modhnfexample}
\end{figure*}

Figure~\ref{fig:modhnfexample} contains an example set of constraints
(generated from code within a polymorphic context without any
monomorphic units defined) encoded into a matrix and stepped through
our modified HNF algorithm. For brevity, the first matrix $H$ is shown
already in HNF but any other arrangement of the rows would have
produced this canonical form. In this first HNF matrix, the problem
with simple HNF can be seen: we started with no units and we still
have no units. Plus, a na\"ive reading of the solution would suggest
that $1.5 x = y$ since $6x = 4y$ (last row). Therefore, the two tasks
for our algorithm are: derive a polymorphic units variable and find
integer solutions for all the unknown units. As we can see, simple HNF
will find integer solutions but they would be expressed in a way that
boils down to $3x=2y$, while what we want looks more like $x=2\alpha$
and $y=3\alpha$. Therefore, modified HNF needs to identify the
situation, derive this polymorphic units variable $\alpha$ that $x$
and $y$ have in common and solve for the two in lowest integer terms.

Our algorithm operates on the HNF $H$ of the input matrix. Each
time we detect a pivot at coordinates $i,j$ that cannot be
scaled down to 1, we take a step to modify matrix $H$ with a new
constraint $i'$ at the bottom and a new column $j'$ on the right:
$H\!\prn{i',j}=1$ and $H\!\prn{i',j'}=-d$ where
\begin{equation*}
  d=\frac{a}{\gcd\set{\sub H{i,j},a}}\quad
  a=\min_{j<k\leq j'}\setc*{\abs{H\!\prn{i,k}}}{\sub H{i,k}\neq 0}
\end{equation*}
This is shown in the Figure~\ref{fig:modhnfexample} example where in
the second matrix the last row and column have been added in this
fashion. We then find the HNF of the modified matrix, producing the
third matrix in the example. Applying two integer elementary row
operations produces the final matrix in reduced row echelon form, from
which the solution can be easily read:
$v=3\alpha, w=12\alpha, x=2\alpha, y=3\alpha$ where $\alpha$ is an
inferred polymorphic units variable represented by the new column.

\newcommand\rank[1]{\mathsf{rank}\!\prn{#1}}

\begin{listing}[t]
  \raggedright
  \textbf{input}: matrix $M$

  let $P_{result}$ be an empty list

  \begin{enumerate}[leftmargin=4em]
  \item[loop (0)] let $H_1=\mathsf{HNF}\!\prn{M}$ and $P$ be an empty list.

  \item for each row $i$ of $H_1$ below $\rank{H_1}$
    \begin{itemize}
    \item let $p_{i,j}$ be the pivot at row $i$,
      column $j$
    \item if $p_{i,j}>1$ and it divides every integer in row
      $i$ then apply that elementary row scaling operation to matrix
      $H_1$
    \item else if $p_{i,j}>1$ but it does not divide every integer in
      the row, then add it to the list $P$.
    \end{itemize}
  \item let $H_2$ be a copy of $H_1$ but only taking $\rank{H_1}$ rows.
  \item for each $p_{i,j}$ in list $P$
    \begin{itemize}
    \item enlarge $H_2$ by one row at index $i'$ and one column at
      index $j'$, with default values being zero.
    \item let {\small $a=\min\setc*{\abs{H\!\prn{i,k}}}{k\in\set{j+1,\ldots,j'}\wedge\sub H{i,k}\neq 0}$}
    \item let $d=a / \gcd\set{p_{i,j},a}$
    \item set $\sub H{i',j}:=1$ and $\sub H{i',j'}:=-d$
    \end{itemize}
  \item let $H_3=\mathsf{HNF}\!\prn{H_2}$
  \item apply integer elementary row scaling and addition operations
    to $H_3$ to change as many pivots as possible to 1 with zeros
    above and below them in the column
  \item let $H_4$ be a copy of $H_3$ but only taking $\rank{H_3}$ rows.
  \item add members of $P$ to $P_{result}$
  \item set $M:=H_4$
  \end{enumerate}
  repeat while $\len{P}>0$

  \textbf{output}: $M$ and $P_{result}$
  \caption{Pseudocode for modified-HNF procedure}
  \label{lst:modhnf}
\end{listing}

The general modified-HNF procedure is shown in pseudocode in
Listing~\ref{lst:modhnf}. The loop body is repeated until no new
columns are generated.

\begin{lemma}\label{lma:pivots}
  For every pivot $p_{i,j}$ detected that does not divide every
  integer in its row, the resulting matrix after the loop body
  completes will have two additional pivots equal to 1. One of them
  will be at column $j$ and one will be at some column greater than
  $j$.
\end{lemma}

\begin{proof}
  The algorithm has generated a new row $i'$ with a 1 in the $j$
  column and the value $d$ in the new $j'$ column.

  Elementary row subtraction of the two rows with the scaling factor
  $\prn{p_{i,j}-1}$ applied to row $i'$ leaves behind a leading 1 and a
  trailing $d\prn{p_{i,j}-1}$ in row $i$.

  Let $a$ be the minimum nonzero absolute value found in row $i$ to
  the right of the pivot, in column $j_a>j$.

  Since $i,j$ is a pivot with value 1, we must subtract row $i$ from
  row $i'$ to cancel the value of 1 we wrote into $i',j$. This leaves
  a value of $-a$ at $\prn{i',j_a}$ and a value of $-dp_{i,j}$ at
  $\prn{i',j'}$.

  We can now obtain the value 1 in row $i'$ and column $j_a$ by applying
  elementary row scaling with a factor of $-a^{-1}$. In order for this
  to maintain integer results the value of $dp_{i,j}$ must be
  divisible by $a$. Therefore when writing the value of $d$ into the
  matrix we selected $d=a/\gcd\set{p_{i,j},a}$ as the minimum such
  integer value.
\end{proof}

\begin{theorem}
  For any matrix given as input, the modified-HNF procedure
  terminates.
\end{theorem}
\begin{proof}
  Let $\len P$ be equal to the number of pivots qualifying under
  Lemma~\ref{lma:pivots} and let $p_{i,j}$ be the right-most such
  pivot. By the lemma, there will be pivots at column $j$ and some
  column $j'>j$, both equal to 1 in the next iteration. The matrix
  then will have $\len P$ more columns but at least $\len P + 1$
  pivots will have been made equal to 1. This reduces the pool of
  potential Lemma~\ref{lma:pivots} candidates by at least one on each
  iteration, until all are exhausted.
\end{proof}

Almost all real-world examples tried so far have been solved in just a few
iterations. The worst case occurs when there are constraints between a
set of variables raised to a sequence of strongly co-prime exponents,
e.g., $x_1^{p_1}+\cdots+x_k^{p_k}$ where $p_1,\ldots,p_k$ are distinct
prime numbers. This type of program is unlikely to be found in the
wild, since it lacks any sensible physical interpretation.  We
artificially generated an example containing terms raised to the first
thirteen prime numbers. That produced a solution for $x_1$ with units
that had to be raised to a power of over 300 trillion. Any more terms
than that soon ran into the upper-bound of 64-bit integers within our tools
and hardware. Even in these artificial cases, the number of loop
iterations was approximately proportional to the number of constraints
between relatively prime co-efficients.

\begin{theorem}
  The modified-HNF procedure outputs only integer matrices.
\end{theorem}
\begin{proof}
  The procedure invokes HNF and integer elementary row operations. It
  only allows division if all the numbers are cleanly divisible. The
  value of $d$ is chosen to be an integer by dividing by the $\gcd$ of
  its numerator's components. All these operations produce integer results
  or matrices, therefore modified-HNF outputs only integer matrices.
\end{proof}

\subsection{Polymorphism and Critical Variables}\label{sec:backendpoly}
In the case of consistent but under-determined systems, some of the
unknown units variables will have many possible solutions. This means
either the provided units annotations are insufficient to give a
single solution, or some variables are units-polymorphic. The
distinction is made based on context (see
Definition~\ref{def:polycontext}): a variable (or formal argument)
declared within a function or subroutine is treated as potentially
units-polymorphic. Outside of a function or subroutine a single units
solution for each variable must be found or else it is considered
underspecified: a candidate `critical variable' that requires manual
annotation. Critical variables are displayed in the output of the
units-suggest mode, as described in Section~\ref{sec:suggest}.

Past work by \citep{Contrastin:2016} has stopped at this point because
all of the consistent uses of polymorphic functions (and subroutines)
will be solved correctly and assigned a single `monomorphic' units
solution in each instance. This gives `implicit' polymorphism:
unannotated variables in functions are quietly treated as
units-polymorphic if possible. In this work, we add support for
`explicit' polymorphism to CamFort, allowing the programmer to write
annotations using an ML-like syntax such as {\tt 'a} or {\tt 'b} %
as shown in Listing~\ref{lst:double}. We also allow the inference
engine to generate explicit specifications automatically for functions
that are implicitly polymorphic, as shown in
Listing~\ref{lst:inferhelper} (left).

\citet{Kennedy:2009:CEFP} and related systems such as \fsharp\ are
extensions of Hindley-Damas-Milner type systems
\citep{Hindley:1969,Damas:1982}, which naturally support inference of
polymorphism and generation of explicit annotations. However, \fsharp\
is stratified (see Section~\ref{sec:fsharp}) therefore it requires a
special and explicit syntax applied to type annotations when units are
involved.  This discourages incremental insertion of units into
existing programs: annotating one variable could require further
changes to be propagated all throughout the program.

Our approach instead assists with \emph{incremental
  specification}. Even with little or no explicit units information
provided, it infers explicit polymorphic units annotations as much as
possible for the program being analysed. One way for that to happen is
through the operation of the modified HNF procedure
(Section~\ref{sec:modhnf}), which produces additional columns that are
categorised as `new base units' representing generated polymorphic
units variables, with an example in
Figure~\ref{fig:modhnfexample}. But even without those additional
columns it is also possible to generate new base units. This occurs
when abstract units variables for function parameters, $\unknownv{\Uabs f k}$, are
underconstrained. We treat some of the unknown abstract units
variables as if they were base units themselves and then solve the
equations in those terms. In other words, we select some of the
non-return unknown abstract unit variables and shift them over to the
right-hand side of the equation. Thus, for each row $\left(c'_1\lhs_1\cdots{}c'_m\lhs_m \middle| b'_1\rhs_1\cdots{}b'_n\rhs_n\right)$
of the solved augmented matrix $H$ from Figure~\ref{fig:matrices},
we rearrange the
row to correspond to the equation/constraint $\sum{\textbf{C}} =
\sum{\textbf{B}_1} - \sum{\textbf{B}_2}$ where:
\begin{align*}
  \setlength{\arraycolsep}{0.05em}
  \begin{array}{lll}
  \textbf{C}&=\setc*{ c'_j \lhs_j}{0<j\leq m\wedge
              \prn{\lhs_j\neq\unknownv{\Uabs f k}\vee k = 0}}
    & \;\;\text{\small{(cols. that are not unknown abstract units)}} \\
    \textbf{B}_1&=\setc*{ b'_j \rhs_j}{0<j\leq n}
    & \;\;\text{\small{(previous RHS)}} \\
  \textbf{B}_2&=\setc*{ c'_j \baseunit{f,k}}{0<j\leq
                m\wedge\lhs_j=\unknownv{\Uabs f k}\wedge k > 0}
    & \;\;\text{\small{(newly generated base units)}}
  \end{array}
\end{align*}
where $\baseunit{f,k}$ denotes a generated `base unit' corresponding
to an implicit units-polymorphic variable. Thus, $\sum{\textbf{B}_2}$
comprises the sum of the newly generated base unit multiplied by the coefficients $c'_j$ of the corresponding
unknown abstract units from the previous left-hand side.

This rearrangement reduces the solution space, ideally to a single point, and
provides a single units assignment for the polymorphic variables that
previously had many possible answers. From that point, all that
remains is presentation: generation of human-readable,
properly-scoped, explicit polymorphic variable names, which is straightforward.

\subsection{Modules and Separate Analysis}\label{sec:fsmod}

Larger programs can potentially generate thousands of constraints or
more, making our system potentially unworkable in real projects unless
the problem is partitioned into manageable chunks (see the quantitative
study in Section~\ref{sec:scale}). Since Fortran 90,
programmers have had the ability to organise their code into modules
that have separate namespaces. Modules contain
variables, subroutines and functions, and may `use' or import
definitions from other modules. The rules for program units
(Figure~\ref{fig:programunitrules}) gather function
and subroutine templates and associate them with modules via
`module-mappings' between module names and the tuple $\prn{\Delta,T}$
where $\Delta$ contains module-variable constraints and $T$ is a
mapping of function/subroutine names to their associated constraint
template. Module-maps can be saved to files known as `fsmod-files' in
our implementation and loaded at a later
time. This imitates a common approach used by Fortran compilers (and many other languages) for
separate compilation.

However, if we were to simply store and re-load every single
constraint found in the templates then we would not reduce the size of
the global constraint-solving problem and would continue to face the
same scalability problems as before. Instead, we follow venerable
data-flow analysis practice \citep[e.g.][]{Barth:1978} and \emph{summarise}
by solving for the units of each function or subroutine parameter in
simplest terms, using our algorithm from Listing~\ref{lst:modhnf} as
described in previous sections. In particular, the ability to infer
explicit polymorphism as described in Section~\ref{sec:backendpoly}
means that we can even reduce the templates of polymorphic functions
to a signature (a minimal set of constraints) that describes the
relationship between their parameters. For example, the \inl{square}
function from Listing~\ref{lst:helpmod} would be described wholly by
the constraint signature $\Uabs f 0\coneq\UEAPabs{\alpha}^2$ and
$\Uabs f 1\coneq\UEAPabs{\alpha}$ where $\alpha$ is a freshly-generated
polymorphic units-variable. Any intermediate constraints regarding the
internals of the function are dropped. This can result in a stark
reduction in the problem size when there are many imports. As an
example, with one of our test-case scientific programs, simulating the
Navier-Stokes equation~\citep{griebel1997numerical}, the main module
that had previously imported approximately 5,000 constraints only
needed 471 constraints when we applied our signature reduction
technique to the data stored in fsmod-files. This leads to roughly a
5x speed up in solving overall (see Section~\ref{sec:scale}).

The module-based solver also allows us to handle another practical
problem: dealing with third party libraries that are supplied without
source code or that cannot be processed with CamFort for any
reason. In such cases the programmer can annotate a stand-in stub
module with the desired units and compile it to an fsmod-file. That file
can be saved or even distributed, and it is loaded by our tool when it
performs inference on a program using the library. Similarly, built-in
Fortran functions are handled by an internal table of
units for each of their parameters (see Section~\ref{sec:intrinsics}).

\subsection{Partial Annotation of Programs}\label{sec:incremental}

Several features of our system combine to enable an incremental
approach to units specification. The \texttt{units-suggest} command
informs the user of a subset of variables to consider giving explicit
annotations. Our annotations are comments, to avoid breaking other
tools and compilers.

Even without annotations, we can infer implicit polymorphic
units and derive explicit specifications for functions and
subroutines. That alone can detect inconsistent usages of variables,
for example an expression such as \inl{sqrt(x) + x} would force our
system to infer the unitless specification \texttt{unit $\unitless$ :: x}, which would be
noticeable when looking over the output of \texttt{units-infer}. Our
handling of variables and literal values makes it easy to introduce
units on one variable and let the global-constraint solver propagate
them as far as possible (see comparison with F\# which does not allow
this in Section~\ref{sec:fsharp}). Separate analysis of modules improves the
performance of the analysis but also helps to organise and
compartmentalise the specification effort.

\subsection{Intrinsics and unit casting}\label{sec:intrinsics}
Fortran language standards specify a base library of arithmetic and trigonometric functions
called `intrinsics'. Our implementation assigns to each of these a suitable polymorphic
unit specification, e.g., \inl{abs} takes an argument with unit \punitA{} and returns
a result of unit \punitA{}. The \inl{transfer} intrinsic performs unsafe
type casts in Fortran, where \inl{transfer(e1, e2)} returns the value of expression \inl{e1}
but at the type of the expression \inl{e2}. We specify \inl{transfer} as taking
 parameters of unit \punitA{} and \punitB{}, returning a result of unit \punitB.
Thus, this may be used for explicit unit casting if using a conversion-factor value (Section \ref{sec:casting}) is not possible.

\section{Scalability}\label{sec:scale}
We provide some quantitative evidence of the scalability of our approach
by timing our tool in inference mode. We algorithmically synthesised
Fortran programs with large numbers of (polymorphic) constraints which
were then generated into two extensionally equivalent code bases
(1) a single file containing many functions; and
(2) multiple files containing a single function, with a top-level module.
By increasing the size of the programs and timing the two classes, we can
observe how our approach supports and promotes modularised code.

\paragraph{Program synthesis}

Algorithmic synthesis of test programs has four parameters
varied as follows:
\begin{enumerate}[leftmargin=1.5em]
\item number of functions, $n \in \{5, 10, 15\}$;
\item length of each function, $l \in \{5, 10, 15, 20\}$;
\item number of function arguments, $a \in \{2, 4\}$;
\item single or multiple-file code base, $\textit{fmt} \in \{\textsf{single}, \textsf{mult}\}$
\end{enumerate}
A function of length $l$ comprises $l$ local variables and $l-2$
assignments of the form \inl{vi = vx * vy} where \inl{vx} is the
$\prn{i-2}$\nth variable and \inl{vy} is the $\prn{i-1}$\nth variable. This
creates a large number of interrelated constraints and tests both
multiplication of units and equality (which has the same effect as constraints
generated by addition). The return result of the function is the $l$\nth
variable. For example, a file generated from parameters
$\textit{fmt} = \textsf{mult}$, $l = 5$, and $a = 2$ is:
\begin{minted}{fortran}
module big1
  real function f1(v1, v2)
    real :: v1, v2, v3, v4, v5
    v3 = (v2 * v1)
    v4 = (v3 * v2)
    v5 = (v4 * v3)
    f1 = v5
  end function f1
end module big1
\end{minted}
\noindent
The parameters of \inl{f1} have units inferred as \texttt{'a :: v1}
and \texttt{'b :: v2} the result as \texttt{('a**2) ('b**3) :: f1}.
The exponents of the polymorphic units variables in the result of
generated programs are the $\prn{l-1}$\nth and $\prn{l-2}$\nth
Fibonacci numbers and each variable \inl{vi} can be expressed as the
multiplication of \inl{v1} and \inl{v0} raised to the $\prn{i-1}$\nth
and $\prn{i-2}$\nth Fibonacci numbers respectively. The fact that
every three consecutive Fibonacci numbers are pairwise co-prime
ensures that intermediate exponents will not be trivially divided away.

The top-level generated program is a subroutine with $l \times n$
parameters which are used to call each function with unique
parameters, but all the units of all the results are unified by
assigning to an additional single, common variable.

\paragraph{Method}

We enumerated the parameter space and recorded the mean time of five
trials for each instance, also calculating the standard error based on
sample variance.  For the multi-file programs, we ran
our tool in \texttt{units-compile} mode to generate intermediate
precompiled files then used \texttt{units-infer} on the
top-level program, importing the precompiled files.  In the single-file mode, all functions are in one
file thus requiring just one invocation of \texttt{units-infer}.

Measurements were taken on a 3.2 Ghz Intel Core i5 machine with 16 GB
 RAM, Mac OS 10.13.4. %

\paragraph{Results}

Table~\ref{tab:scale-results} gives the results. It is clear that the
single file programs suffer from non-linear scaling (somewhere between
quadratic and cubic time), whereas multiple-file compilation scales
well. Our approach thus successfully mitigates the high-cost of global constraint
 solving for units-of-measure in highly polymorphic contexts (which is
 expensive, as shown by the single file case).

For comparison, Table~\ref{tab:scale-results2} gives the inference
time for two real numerical models: a
Navier-Stokes fluid model~\citep{griebel1997numerical},
split across six modules totalling 517 physical lines of
code; and \emph{cliffs} a tsunami model comprising 2.7kloc over 30
modules~\citep{tolkova2014land}.\footnote{Source line counts are of
  physical lines (excluding comments and whitespace) generated by
  SLOCCOUNT~\citep{wheeler2001sloccount}.} We also include
inference times for the case study modules considered in
Section~\ref{sec:studies} taken from \emph{wrf}, the Weather Research and Forecast Model~\citep{Skamarock:2008:WRF}.

Navier-Stokes is comparable in size
to the $n = 15, l = 15, a=2$ programs (514 physical lines of code
in single file and 573 lines of code in 15 files).  We also
coalesced the program units of the Navier code into a single file
(\emph{navier-alt}) to compare with the single file performance on similar-sized generated files. The
constraints generated from the real-world program are clearly less
onerous in the single-file case than our generated programs.%

\newcommand{\errr}[1]{\!\!\!{\small{$\pm$#1}}}
\newcommand{\sloc}[1]{\textcolor{darkgrey}{\small{#1}}}
\begin{table}
  \begin{tabular}{c|c||p{1.0em}S[table-format=4.3]l|p{1.0em}S[table-format=1.3]l||p{1.0em}S[table-format=3.3]l|p{1.0em}S[table-format=2.3]l}
 & &  \multicolumn{6}{c||}{$a = 2$} & \multicolumn{6}{c}{$a = 4$} \\
n  & l  & loc & \multicolumn{2}{c}{\textsf{single} (s)} &  loc &
 \multicolumn{2}{c||}{\textsf{mult}  (s)}
 & loc & \multicolumn{2}{c}{\textsf{single} (s)} & loc &
 \multicolumn{2}{c}{\textsf{mult} (s)} \\
 \hline \hline
5  &  5  & \sloc{74} & 0.649 & \errr{0.002} & \sloc{93} & 0.557 & \errr{0.005} &
\sloc{84} & 0.790  & \errr{0.004} & \sloc{103} & 0.687  & \errr{0.013} \\
 5  & 10 & \sloc{124} & 1.690 & \errr{0.003} & \sloc{143} & 0.883 &  \errr{0.001} &
\sloc{134} & 2.026  & \errr{0.003} & \sloc{153} & 1.228  & \errr{0.024} \\
 5  & 15 & \sloc{174} & 3.681 & \errr{0.003} & \sloc{193} & 1.200 & \errr{0.003} &
\sloc{184} & 4.640  & \errr{0.003} & \sloc{203} & 1.696  & \errr{0.030} \\
 5  & 20 & \sloc{224} & 7.504 & \errr{0.003} & \sloc{243} & 1.513 & \errr{0.008} &
\sloc{234} & 9.959  & \errr{0.006} & \sloc{253} & 2.234  & \errr{0.041} \\
10  &  5 & \sloc{144} & 2.141 & \errr{0.002} & \sloc{183} & 1.254 & \errr{0.013} &
\sloc{164} & 1.265  & \errr{0.002} & \sloc{203} & 1.675  & \errr{0.028} \\
10  & 10 & \sloc{244} & 6.979 & \errr{0.003} & \sloc{283} & 1.915 & \errr{0.012} &
\sloc{264} & 21.024  & \errr{0.010} & \sloc{303} & 5.714  & \errr{0.159} \\
10  & 15 & \sloc{344} & 56.492 & \errr{0.127} & \sloc{383} & 3.364 & \errr{0.040} &
\sloc{364} & 103.297  & \errr{0.048} & \sloc{403} & 12.224  & \errr{0.386} \\
10  & 20 & \sloc{444} & 41.853 & \errr{0.232} & \sloc{483} & 3.658 & \errr{0.023} &
\sloc{464} & 142.871  & \errr{0.041} & \sloc{503} & 14.926  & \errr{0.463} \\
15  &  5 & \sloc{214} & 7.692 & \errr{0.005} & \sloc{273} & 2.074 & \errr{0.019} &
\sloc{244} & 2.542  & \errr{0.007} & \sloc{303} & 3.972  & \errr{0.069} \\
15  & 10 & \sloc{364} & 38.043 & \errr{0.006} & \sloc{423} & 3.915 & \errr{0.040} &
\sloc{394} & 50.524  & \errr{0.080} & \sloc{453} & 35.471  & \errr{0.898} \\
15  & 15 & \sloc{514}  & 205.565 & \errr{0.059} & \sloc{573} & 6.014 & \errr{0.070} &
\sloc{544} & 284.832  & \errr{0.107} & \sloc{603} & 99.916  & \errr{2.624} \\
15 & 20  & \sloc{664} & 1009.508 & \errr{0.417} & \sloc{723} & 8.496  & \errr{0.107} &
\sloc{694} & 456.948  & \errr{0.120} & \sloc{753} & 81.796  & \errr{2.101}
  \end{tabular}
\vspace{0.6em}
\caption{Timing results for scalability experiments
 (in seconds, to 3 s.f. with standard (sample) error)}
\label{tab:scale-results}
\end{table}

\begin{table}
\vspace{-1.2em}
\begin{tabular}{r|S[table-format=4]|c||S[table-format=3.3]l}
package & loc & files & \multicolumn{2}{c}{\textsf{inference time}(s)} \\ \hline\hline
\textit{navier} & 517 & 6 & 10.241 & \errr{0.303} \\
\textit{navier-alt} & 481 & 1 & 47.704 & \errr{0.002} \\
\textit{cliffs} & 2706 & 30 & 40.881 & \errr{0.121} \\ \hline
\emph{wrf} & \multicolumn{4}{l}{\textit{}} \\ \hline
  (chem) \textsf{aerorate\_so2} & 82 & 3 & 0.370 & \errr{0.002} \\
  (phys) \textsf{module\_sf\_oml} & 127 & 1 & 3.361 & \errr{0.002} \\
  (phys) \textsf{module\_cam\_wv\_saturation} & 693 & 1 & 20.685 & \errr{0.008}
\end{tabular}
\vspace{0.6em}
\caption{Timing results on some real packages
(in seconds, to 3 s.f. with standard (sample) error)}
\label{tab:scale-results2}
\vspace{-1em}
\end{table}

\section{Case Studies}\label{sec:studies}
One difficulty of units analysis compared to some other forms of static analysis is that it requires knowledge of programmer intention and domain in order to provide meaningful annotations. Therefore the process can only be partly automated. With existing bodies of code we rely upon documentation in order to provide information about the units that should be assigned to variables. Unfortunately for most real scientific programs that documentation is often incomplete, if present at all. For samples of Fortran code that are readily available, widely used and mature we took the Weather Research and Forecast Model~\citep{Skamarock:2008:WRF} by the National Center for Atmospheric Research. This project contains over half a million lines of Fortran in about 200 files with numerous add-on sub-projects as well. The code has been developed over many years and is of varying quality. Some files have comments describing the units-of-measure for many variables. We focused on a selection of those files and worked on annotating the variables based on their accompanying comments. We filled in the remaining gaps using an iterative process of inference and consultation of scientific textbooks and papers from which the algorithms had been derived.
\begin{itemize}[leftmargin=1.5em]
\item {\bf aerorate\_so2}: this file is found in the Chem subproject
and it depends upon two other modules. %
We applied 7 units annotations within the file and 11 annotations within the other modules. This process led us to finding an inconsistency in one equation where the program was attempting to convert a value from grams to micro-grams by multiplying by $10^{-6}$ instead of $10^6$.
\item {\bf module\_sf\_oml}: this module is from the core physics code comprising 127 physical lines of code. It is an implementation of an old and lesser-used mixed ocean layer model~\citep{Pollard:1973}. We added 25 annotations to this file for documentation purposes but at least a third of them could have been omitted for inference purposes. As we narrowed down the assignment of units we found a contradiction for some of the helper variables holding intermediate results for wind speed. We have provided feedback to the developers about this inconsistency.
\item {\bf module\_cam\_wv\_saturation}: this file is from the core physics code and has 693 physical lines of code. It has routines that estimate saturation vapour pressure and specific humidity. We added 28 annotations to this file and found an `inconsistent use of units' in several instances of an equation. However, this equation computes a derivative of pressure over temperature for specific ranges of inputs. One term used a change-in-Kelvin while the other used a change-in-Celsius. In this case, these are compatible quantities. In general we would advise users to avoid mixing Kelvin and Celsius because in other contexts they are not immediately interchangeable.
Thus, whilst not a direct bug, our tool indicated undocumented unit
conversion behaviour.

\end{itemize}
 
\section{Related Work and Discussion}\label{sec:related}
\paragraph{Kennedy and \fsharp}\label{sec:fsharp}

The most mainstream implementation of \citet{Kennedy:1996} is found in a
standard extension of the \fsharp\ language~\citep{Kennedy:2009:CEFP}. Units of
measure may be declared or aliased using a special syntax:
\begin{minted}[xleftmargin=0.2em]{fsharp}
[<Measure>] type cm
[<Measure>] type ml = cm^3
\end{minted}
Types and literal values are annotated directly using a syntactic extension:
\begin{minted}[xleftmargin=0.2em]{fsharp}
let cm_to_inch(x:float<cm>) = x / 2.54<cm/inch>
\end{minted}
\noindent
The types are stratified: the \inlFS{float} type cannot be unified
with a units-annotated \inlFS{float<u>} type, only with
\inlFS{float<1>}. This makes transitioning existing code-bases
difficult because adding units annotations to a function means that
you have to annotate all of its call-sites as well. Our approach
requires as little annotation as possible, using global
constraint-solving to fill in the gaps. %

In Kennedy's inference algorithm, nonzero literals are inferred as
unitless unless given an explicit annotation.  Using our notation (and
our representation of abstract units variables), Kennedy's rules for
literals could be written as:
\begin{align*}
  \erule{Kpolyzero}{\ujudge{\Gamma}{0}{\Uabs f l}{\cdot}}{\freshlitid l}
  \qquad
  \erule{Kliteral}{\ujudge{\Gamma}{n}{\unitless}{\cdot}}{n \neq 0}
\end{align*}
That is, $0$ is treated polymorphically (with a fresh abstract unit
variable) and all other unannotated literals are treated as unitless.  However, we
distinguish literals at the top-level of a program/module from those
in a function/subroutine: see \erulenamed{literal} from
Figure~\ref{fig:literalrules}.
Thus, units constraints provided through the literal exception (see
Section~\ref{sec:literals}) can provide a bridge between the units of
other expressions, without requiring all literals to be annotated. The
effect can be seen in the comparison between \fsharp\ and our
extension of CamFort shown in Listing~\ref{lst:fsharpcamfort}.  Fixing
the \fsharp\ example would require the literal value \inlFS{4.0} to be
explicitly marked with either the annotation \inlFS{<kg>} or the wildcard
\inlFS{<_>}.

\begin{listing}[t]
{\raggedright
\begin{minipage}{0.45\linewidth}
\begin{minted}[xleftmargin=0.2em]{fsharp}
// F#
let x = 3.0<kg>
let y = 4.0
let z = x + y
// Error:
// units mismatch with x and y
\end{minted}
\end{minipage}
\begin{minipage}{0.45\linewidth}
\begin{minted}[xleftmargin=0.2em]{Fortran}
! CamFort extension
!= unit kg :: x
real :: x = 3.0, y = 4.0, z
z = x + y
! OK:
! inferred unit kg :: y, z
\end{minted}
\end{minipage}
}
\caption{Annotating a single literal is sufficient for CamFort.}\label{lst:fsharpcamfort}
\end{listing}

Furthermore, \fsharp\ units are annotated directly on variables and
literals, meaning that any tools operating on the \fsharp\ code must
understand the units syntax. The CamFort approach, which we share, has
been to avoid custom code constructs and instead use a comment-based
annotation syntax that can be safely ignored by any tools (and
compilers) that do not know about units. The goal is that CamFort's
units annotations can safely be ignored or erased without affecting
program operation. That design choice leads to certain consequences,
including that literals cannot be annotated directly, and so we do not
support any notion of a custom `units-casting operation' (although
Fortran's casting operation \inl{transfer} could be used this
way). Instead, we allow a literal to take on a variable's explicitly
annotated units if it is the sole expression in an assignment to that
variable (Section~\ref{sec:expressions}).

\paragraph{Phriky Units}
The approach taken by \citet{Ore:2017:ISSTA} avoids the
need for any units-of-measure annotation within C++
programs. This is achieved by first creating files containing sets of
annotations for commonly shared libraries, and then propagating that
units information outward from any points in programs that use those
shared libraries. Because the annotations are associated with shared
libraries, user-programs need not be annotated themselves if they
use the annotated library. The algorithm for checking units
is intentionally unsound for simplicity and performance purposes. It
employs a form of intraprocedural dataflow analysis that approximates
the true set of constraints, sometimes too loosely and sometimes too
tightly. While this can lead to clear problems it also means that the
system is flexible in the presence of C++ code structures that are not
fully analysable by the tool. Instead of giving polymorphic units to
functions, Phriky Units simply skips these functions if there are no monomorphic
units constraints to check.

\paragraph{Osprey}

\citet{Jiang:2006:ICSE} developed Osprey, a scalable units-of-measure
checking system for the C programming language. Like ours, their
system supports suggesting critical variables and units aliasing. But
they only have support for \textit{leaf polymorphism}, which is a
limited form of generic units mechanism that only works within the
function calls at the leaves of the call-tree. In contrast, we can
infer specifications for polymorphic units variables at arbitrary
depths; meaning that polymorphic functions can call polymorphic
functions in our system, important for larger functions that may be
(partly) polymorphic. Additionally, Osprey allows compositional units
analysis only at function boundaries while we also support separate
analysis that can use the modular structure of the program to
partition the work.

\paragraph{CamFort}

The problem of adding units to Fortran programs has more recently been
studied by \citet{Orchard:2015:JCS, Contrastin:2016} as part of the
CamFort project~\citep{CamFort}. Their work was similar to Osprey in
its theoretical capability for analysing units, except that CamFort fully
supported implicit units-polymorphism at all levels of the
call-tree. Their system would try to assume that a
function was units-polymorphic if it had no explicit units
annotations. We found that the CamFort units analysis engine was not
practical for use on real programs due to serious inefficiency and lack of
module support. We have reused and adapted some of the open
source CamFort code for parsing and analysing Fortran code. On top of
that, we have largely rewritten and greatly extended the
units-inference and checking capabilities of CamFort to support
important features that we describe in this paper such as checking and
inferring explicit annotations for units-polymorphism, working across
modules within larger projects and an enhanced and much more efficient
solver based on our modified HNF algorithm (Section~\ref{sec:modhnf}).

\paragraph{N1969}

\citet{Snyder:2013:N1969} officially proposed an ISO standard
extension of Fortran syntax, now rejected, to accommodate the
specification of units-of-measure and type unification supporting the
equational theory of Abelian groups. This included two new statements
for declaring the names of units and a units attribute that could be
provided for variable declarations. The original CamFort units syntax
was based on this proposal~\citep{Orchard:2015:JCS} however it has
since evolved to a comment-based annotation syntax based on the
principle of \emph{harmlessness} to other tools
(Section~\ref{sec:harmlessness}). N1969 did not propose inference and
would require units annotation of every variable, which renders it
impractical, especially for adoption in existing projects.

\begin{listing}[b]
\hspace{-3em}\begin{minipage}{0.45\linewidth}
\begin{minted}[xleftmargin=0.2em]{C}
//@ post(UNITS): @unit(@result) = @unit(x)
int sqr(int x) {
  int y = 0, i;
  for(i=0;i<x;i+=1)
    y += x;
  return y;
}
\end{minted}
\end{minipage}
\begin{minipage}{0.45\linewidth}
\begin{minted}[xleftmargin=0.2em]{C}
//@ post(UNITS): @unit(@result) = @unit(x)
double f(double x) {
  return x + 2.0;
}
// e.g. f(2.54 * 1.0) != 2.54 * f(1.0)
\end{minted}
\end{minipage}
\caption{Unsound examples in CPF[UNITS]: a function \mintinline{C}{sqr} that squares its parameter but not its units (left), and a function \mintinline{C}{f} that breaks dimensional invariance (see Section~\ref{sec:monorestrict}) (right).}
  \label{lst:cpfunits}
\end{listing}

\paragraph{CPF[UNITS]}
\citet{Hills:2008:RULE,Chen:2003:RTA,Rosu:2003:ASE} describe a
comment-based syntax for units-of-measure annotation on the \emph{bc
calculator} and C languages. These rely upon the Maude system of
rewriting logic.  Units can be checked dynamically or statically; we
focus on the latter since it is comparable to our system:

\begin{itemize}[leftmargin=1.5em]
\item Variables can be assigned different units in different program
  branches, unlike a typical type system, therefore static checking
  must consider multiple possible assignments of units to
  variables, risking exponential running time for static verification. Loops
  must be re-analysed until they reach a fixed point, or handled by
  some predefined `code pattern' matching mechanism.
\item It is possible to describe relationships between a function's
  input parameters and its output results, e.g. %
  {\tt @unit(@result) = @unit(parameter)}, %
  and that provides a facility that is similar to
  units-polymorphism for the purpose of specifying function behaviour.
\item There is no units-inference or critical variable suggestion
  facility, only checking.
\item One of the most important assumptions of the CPF[UNITS] is called
  the `locality principle', which says when the programmer uses a
  constant with a simple instruction such as $x+5$ the checker assumes
  that the constant has the same units as the variable. This is
  superficially similar to our literal exception
  (Section~\ref{sec:literals}), but we do not allow any other
  operation than direct assignment of a literal value to a variable,
  and polymorphism is not allowed. The `locality principle' breaks
  soundness in CPF[UNITS] by allowing you to write units-polymorphic
  functions that are not dimensionally invariant (for more details see
  Section~\ref{sec:monorestrict}). In Listing~\ref{lst:cpfunits} we
  show two examples that exploit the problems caused by the locality
  principle. First we define a function \mintinline{C}{sqr} that
  numerically returns the square of its input, but dimensionally does
  not. Second we define a function \mintinline{C}{f} that breaks
  dimensional invariance. These two cases happen because the system
  allows us to treat the literals \mintinline{C}{1} and
  \mintinline{C}{2.0} respectively as having polymorphic units. Since
  constraints are compiled into Maude, it is not clear how this
  problem can be fixed, and CPF[UNITS] is no longer an active project.
\end{itemize}

\paragraph{Squants}
\citet{Squants} have released a framework for Scala that provides
units-of-measure functionality within Scala's expressive type
system. Squants provides a hierarchy of dimensions with associated
units and easy insertion into Scala programs with many conveniences
such as easy conversion and static checking.

However, this system is not capable of solving equations on units as
an Abelian group, therefore it is limited in expressivity. For
example, you cannot multiply together arbitrary quantities with dimensions; only predefined or user-defined combination of units are
supported. It is impossible to define functions that behave like our
example \mintinline{Fortran}{square} in Listing~\ref{lst:square}. The
abstract class \mintinline{Scala}{Quantity} only supports
multiplication by a dimensionless \mintinline{Scala}{Double},
therefore the function \mintinline{Scala}{sqr} in
Listing~\ref{lst:squants} squares the value but merely returns the
same units that it is given.
Still, given this limitation, the authors have defined many convenient
units. For example, if you multiply \mintinline{Scala}{Meters(2)} by
\mintinline{Scala}{Meters(2)} the result is
\mintinline{Scala}{SquareMeters(4)}. And if you multiply
\mintinline{Scala}{SquareMeters(4)} by \mintinline{Scala}{Meters(2)}
then the result is \mintinline{Scala}{CubicMeters(8)}. However if you
multiply \mintinline{Scala}{CubicMeters(8)} by
\mintinline{Scala}{Meters(2)} then the result is an error.

The addition operation of abstract class \mintinline{Scala}{Quantity}
operates on two values of the same class, and it is possible to
polymorphically create a literal having the same units as a given
parameter. This is shown in Listing~\ref{lst:squants} by
function~\mintinline{Scala}{f} that uses the same technique as
Listing~\ref{lst:cpfunits} to break dimensional invariance (see
Section~\ref{sec:monorestrict} for details). Examples executed in the
Scala REPL show the unfortunate results. With Squants we can define a function \mintinline{Scala}{f} for which
\mintinline{Scala}{f(Inches(1)) in Centimeters} produces a different
result than \mintinline{Scala}{f(Inches(1) in Centimeters)}, therefore
dimensional invariance is broken.

\begin{listing}[t]
\begin{minted}[xleftmargin=0.2em]{Scala}
def sqr[A <: Quantity[A]](x: Quantity[A]): Quantity[A] = {x*x.value}
def f[A <: Quantity[A]](x: Quantity[A]): Quantity[A] = {x+x.unit.apply(2)}

scala> sqr(Meters(3))
res0: squants.Quantity[squants.space.Length] = 9.0 m
scala> f(Inches(1)) in Centimeters
res1: squants.space.Length = 7.62001524 cm
scala> f(Inches(1) in Centimeters)
res2: squants.Quantity[squants.space.Length] = 4.54000508 cm
\end{minted}

\caption{Examples in Squants and the Scala REPL: a function \mintinline{Scala}{sqr} that squares its parameter but not its units, and a function \mintinline{Scala}{f} that breaks dimensional invariance (see Section~\ref{sec:monorestrict}).}
  \label{lst:squants}
\end{listing}

\paragraph{Other work integrated with static types}
There are several other libraries that also take advantage of their
language's advanced type system features to provide a
statically-checked units-of-measure specification feature. C++ has
Boost::Units~\citep{Boost:Units}, the package `dimensional' is
available for Haskell~\citep{Dimensional} and
JScience~\citep{JScience} is a scientific programming package for Java
that includes a units-of-measure library. All of these offer
statically-checked units that integrate with the host
language. However that comes with a downside, because they cannot
easily embed the equational theory of an Abelian group, the uses of
units can be somewhat awkward and units-inference is not possible in
the same way that \fsharp\ or our work can manage.

\citet{gundry2015typechecker} provides a units-of-measure system
for Haskell, reusing Haskell's type inference algorithm with a custom
solver via a type-checker plugin, and macros to extend the surface
syntax. Since Gundry's approach reuses type inference, the constraint
generation is similar to that of F\#'s.

A notable past effort includes the Fortress
language~\citep{Allen:2005:Fortress} that was designed to include
units-of-measure from the very start, however work on the language has
been discontinued. Also, the author of JScience attempted to have a
units specification added to the Java language but the proposal was
rejected~\citep{JSR275}.

\paragraph{Semantics}
\citet{atkey2015models} introduce a simple calculus with
polymorphic units for the purposes of modelling the underlying type
theory. They do not present any algorithms for checking or inference,
but instead give a type system specification (via inference rules)
with introduction and elimination rules for universal quantification
over units-of-measure variables.  Their semantics gives a basis for
proving parametricity theorems on units-polymorphic functions.

\paragraph{Run-time checking of units-of-measure}
Several systems and libraries offer tagging of numbers with
units-of-measure `metadata' and dynamically checking the results of
computations to ensure that the units are consistent. These include
several Python modules \citep{Pint,numericalunits,units:Py}, an R
library \citep{units:R}, a Haskell package \citep{quantities}, a
Fortran library \citep{Petty:2001} and many others that work in a
similar manner. These systems take a fundamentally different approach
to ours, by only catching errors after they have happened, while our
work seeks to find inconsistencies before they strike and with zero
run-time overhead. The benefits of run-time checking are an easier
implementation and more flexibility. The costs are overhead in memory
and time, the need to handle the resulting errors at run-time and the
possibility of missing bugs in programs that are incompletely tested.

\subsection{Conclusion and further work}\label{sec:futurework}
We have described a new system for incremental inference and
verification of units-of-measure on modular programs. We have
created an algorithm based on Hermite Normal Form to solve constraint
problems involving polymorphic units, and to produce signatures for
functions and subroutines within modules. We make it possible to
partition the global constraint solving problem into smaller pieces
that are more easily solved, and show that this approach has great
potential to speed-up the analysis. Our method of annotation maintains
the principles of harmlessness and incrementalism by relying solely on
a comment-based syntax and allowing users to write annotations at
their own pace, working with whatever is provided the best we can. We
can also automate part of the work with suggestions and code
synthesis. We believe that our system addresses many of the problems
that have beset past work in software units and dimensional analysis,
and we intend to further refine it into a practical tool for the
working scientist.

During this project,
we have noticed that some scientific formulas, at least for
intermediate values, involve taking logarithms of units
\citep[e.g. ][]{GoffGratch:1946}, or raising units to rational powers,
such as with the Gauckler-Manning coefficient
\citep[p79]{Chanson:2004}. Further work is to find a way to sensibly
encode these formulas without disturbing the rest of the system.

Another area for exploration is array indices:
units-annotated integer variables may be used as array indices but we
do not currently have support for specifying the units of indices.

\begin{acks}                            %
  This work has been supported by
  Grant {EP/M026124/1} from the {Engineering and Physical Sciences Research Council}.%
  Any opinions, findings, and conclusions or recommendations expressed in
  this material are those of the authors and do not necessarily
  reflect the views of the EPSRC.
\end{acks}

\bibliography{references}

\end{document}